\documentclass[11pt]{article}

\usepackage[utf8]{inputenc}
\usepackage[T1]{fontenc}
\usepackage{lmodern}
\usepackage[DIV=11]{typearea} 

\usepackage{microtype}

\usepackage{amssymb}
\usepackage{amsmath}
\usepackage{amsthm}
\usepackage{thmtools}

\usepackage[ruled,vlined,linesnumbered,nokwfunc]{algorithm2e_}

\usepackage{tikz}

\usepackage{xcolor}
\usepackage{xspace}
\usepackage{paralist}
\usepackage{comment}

\usepackage[
	style=alphabetic,
	backref=true,
	doi=false,
	url=false,
	maxcitenames=3,
	mincitenames=3,
	maxbibnames=10,
	minbibnames=10,
	backend=bibtex8,
	sortlocale=en_US
]{biblatex}

\usepackage[ocgcolorlinks]{hyperref} 

\newcommand*\samethanks[1][\value{footnote}]{\footnotemark[#1]}

\makeatletter
\g@addto@macro\bfseries{\boldmath}
\makeatother

\makeatletter
\g@addto@macro\mdseries{\unboldmath}
\g@addto@macro\normalfont{\unboldmath}
\g@addto@macro\rmfamily{\unboldmath}
\g@addto@macro\upshape{\unboldmath}
\makeatother

\DeclareCiteCommand{\citem}
    {}
    {\mkbibbrackets{\bibhyperref{\usebibmacro{postnote}}}}
    {\multicitedelim}
    {}

\renewcommand*{\multicitedelim}{\addcomma\space}

\AtEveryCitekey{\clearfield{note}}

\newcommand{\myhref}[1]{%
  \iffieldundef{doi}
    {\iffieldundef{url}
       {#1}
       {\href{\strfield{url}}{#1}}}
    {\href{http://dx.doi.org/\strfield{doi}}{#1}}%
}

\DeclareFieldFormat{title}{\myhref{\mkbibemph{#1}}}
\DeclareFieldFormat
  [article,inbook,incollection,inproceedings,patent,thesis,unpublished]
  {title}{\myhref{\mkbibquote{#1\isdot}}}

\makeatletter
\AtBeginDocument{%
    \newlength{\temp@x}%
    \newlength{\temp@y}%
    \newlength{\temp@w}%
    \newlength{\temp@h}%
    \def\my@coords#1#2#3#4{%
      \setlength{\temp@x}{#1}%
      \setlength{\temp@y}{#2}%
      \setlength{\temp@w}{#3}%
      \setlength{\temp@h}{#4}%
      \adjustlengths{}%
      \my@pdfliteral{\strip@pt\temp@x\space\strip@pt\temp@y\space\strip@pt\temp@w\space\strip@pt\temp@h\space re}}%
    \ifpdf
      \typeout{In PDF mode}%
      \def\my@pdfliteral#1{\pdfliteral page{#1}}
      \def\adjustlengths{}%
    \fi
    \ifxetex
      \def\my@pdfliteral #1{}
      \def\adjustlengths{\setlength{\temp@h}{-\temp@h}\addtolength{\temp@y}{1in}\addtolength{\temp@x}{-1in}}%
    \fi%
    \def\Hy@colorlink#1{%
      \begingroup
        \ifHy@ocgcolorlinks
          \def\Hy@ocgcolor{#1}%
          \my@pdfliteral{q}%
          \my@pdfliteral{7 Tr}
        \else
          \HyColor@UseColor#1%
        \fi
    }%
    \def\Hy@endcolorlink{%
      \ifHy@ocgcolorlinks%
        \my@pdfliteral{/OC/OCPrint BDC}%
        \my@coords{0pt}{0pt}{\pdfpagewidth}{\pdfpageheight}%
        \my@pdfliteral{F}
        \my@pdfliteral{EMC/OC/OCView BDC}%
        \begingroup%
          \expandafter\HyColor@UseColor\Hy@ocgcolor%
          \my@coords{0pt}{0pt}{\pdfpagewidth}{\pdfpageheight}%
          \my@pdfliteral{F}
        \endgroup%
        \my@pdfliteral{EMC}%
        \my@pdfliteral{0 Tr}
        \my@pdfliteral{Q}%
      \fi
      \endgroup
    }%
}
\makeatother

\newcommand{\set}[1]{\{#1\}}
\newcommand{\lu}{\textup{(}}
\newcommand{\ru}{\textup{)}\xspace}
\newcommand{\enum}[1]{\lu #1\ru}

\newcommand{\tscc}{\text{tSCC}\xspace}
\newcommand{\bscc}{\text{bSCC}\xspace}
\newcommand{\scc}{\text{SCC}\xspace}
\newcommand{\vscc}{\text{2vSCC}\xspace}
\newcommand{\escc}{\text{2eSCC}\xspace}

\newcommand{\kscc}{\text{kSCC}\xspace}
\newcommand{\kvscc}{\text{kvSCC}\xspace}
\newcommand{\kescc}{\text{keSCC}\xspace}
\newcommand{\nt}[3]{\ensuremath{\textrm{$k$-almostTSCC}(#1, #2, #3)}\xspace}

\newcommand{\revG}[1]{\mathit{Rev({#1})}\xspace}
\newcommand{\rev}[1]{{#1}^R\xspace}
\newcommand{\G}{\mathcal{G}\xspace}
\newcommand{\bad}[2]{B_{{#1},{#2}}\xspace}
\newcommand{\good}[2]{A_{{#1},{#2}}\xspace}
\newcommand{\subG}[2]{#1_{#2}\xspace}
\newcommand{\subV}[1]{V_{#1}\xspace}
\newcommand{\subE}[1]{E_{#1}\xspace}
\newcommand{\subVf}[1]{V'_{#1}\xspace}
\newcommand{\subEf}[1]{E'_{#1}\xspace}
\newcommand{\subGf}[2]{F_{{#1},{#2}}\xspace}
\newcommand{\subGe}[2]{\subGf{#1}{#2}}
\newcommand{\subGv}[2]{\subGf{#1}{#2}}
\newcommand{\rootf}[2]{r_{{#1},{#2}}\xspace}
\newcommand{\roote}[2]{\rootf{#1}{#2}}
\newcommand{\rootv}[2]{\rootf{#1}{#2}}
\newcommand{\GS}{G_{S,Z}\xspace}
\newcommand{\GV}{G_{V \setminus S}\xspace}
\newcommand{\jj}{j\xspace}
\newcommand{\mmdeg}{\gamma}

\newcommand{\lost}{q}
\newcommand{\depth}{d}
\newcommand{\maxdeg}[1]{t_{#1}\xspace}
\newcommand{\subVG}[2]{V_{{#1},{#2}}\xspace}

\DeclareMathOperator{\Out}{Out}
\DeclareMathOperator{\In}{In}
\DeclareMathOperator{\OutDeg}{Outdeg}
\DeclareMathOperator{\InDeg}{Indeg}

\DeclareMathOperator{\topsccex}{\ProcNameSty{TopSCCWithout}} 
\DeclareMathOperator{\topscc}{\ProcNameSty{TopSCC}}

\DeclareMathOperator{\allbridges}{\ProcNameSty{Bridges}}

\DeclareMathOperator{\localalg}{\ProcNameSty{2IsolatedSetLocal}}
\DeclareMathOperator{\levelsearch}{\ProcNameSty{2IsolatedSetLevel}}
\DeclareMathOperator{\klevelsearch}{\ProcNameSty{kIsolatedSetLevel}}
\DeclareMathOperator{\allsearch}{\ProcNameSty{2IsolatedSet}}
\DeclareMathOperator{\kallsearch}{\ProcNameSty{kIsolatedSet}}
\DeclareMathOperator{\bfs}{\ProcNameSty{BFS}}
\DeclareMathOperator{\rec}{\ProcNameSty{\vscc}}
\DeclareMathOperator{\reck}{\ProcNameSty{\kscc}}

\newif\iffullversion
\newcommand{\infull}[1]{\iffullversion #1\fi}
\newcommand{\inshort}[1]{\iffullversion \else #1\fi}

\fullversiontrue

\usetikzlibrary{arrows,shapes}

\colorlet{DarkRed}{red!50!black}
\colorlet{DarkGreen}{green!50!black}
\colorlet{DarkBlue}{blue!50!black}

\hypersetup{
	linkcolor = DarkRed,
	citecolor = DarkGreen,
	urlcolor = DarkBlue,
	bookmarks = true,
	bookmarksnumbered = true,
	linktocpage = true
}

\declaretheorem[numberwithin=section]{theorem}
\declaretheorem[numberlike=theorem]{lemma}

\declaretheorem[numberlike=theorem]{corollary}
\declaretheorem[numberlike=theorem]{definition}
\declaretheorem[numberlike=theorem]{observation}

\DontPrintSemicolon
\SetAlCapFnt{\normalfont\scshape}
\SetProcNameSty{texttt}
\SetFuncSty{texttt}

\title{Finding 2-Edge and 2-Vertex Strongly Connected Components in Quadratic Time\thanks{A preliminary version of this paper is presented at the \emph{42nd International Colloquium on Automata, Languages, and Programming (ICALP 2015)}.}}

\author{
Monika Henzinger\thanks{University of Vienna, Faculty of Computer Science, Austria. This work was supported by the Austrian Science Fund (FWF): P23499-N23. Additionally, the research leading to these results has received funding from the European Research Council under the European Union's Seventh Framework Programme (FP/2007-2013) / ERC Grant Agreement no. 340506.}
\and Sebastian Krinninger\samethanks[2]
\and Veronika Loitzenbauer\samethanks[2]
}

\date{}

\hypersetup{
	pdftitle = {Finding 2-Edge and 2-Vertex Strongly Connected Components in Quadratic Time},
	pdfauthor = {Monika Henzinger, Sebastian Krinninger, Veronika Loitzenbauer}
}

\begin{document}
\maketitle

\begin{abstract}
We present faster algorithms for computing the 2-edge and 2-vertex strongly connected components of a directed graph. While in \emph{undirected} graphs the 
2-edge and 2-vertex connected components can be 
found in linear time, in \emph{directed} graphs with $m$ edges and $n$ vertices
only rather simple 
$O(m n)$-time algorithms were known. We use a hierarchical sparsification technique to obtain algorithms
that run in time $O(n^2)$.
For 2-edge strongly connected components our algorithm gives the first running time improvement in 20 years.
Additionally we present an 
$O(m^2 /\log n)$-time algorithm for 2-edge strongly connected components, and 
thus improve over the $O(m n)$ running time also when $m = O(n)$.
Our approach extends to $k$-edge and $k$-vertex strongly connected components
for any constant $k$ with a running time of $O(n^2 \log n)$ for 
$k$-edge-connectivity and $O(n^3)$ for $k$-vertex-connectivity.

\end{abstract}

\section{Introduction}\label{sec:intro}

\paragraph*{Problem Description.}
In a directed graph~$G$ two vertices $u$ and $v$ are \emph{2-edge strongly connected}
if from $u$ to $v$ and from $v$ to $u$, respectively, there are two paths that 
have no common edge. A \emph{2-edge strongly connected component} (\escc)
of~$G$ is a maximal subgraph of~$G$ such that in the subgraph every pair of distinct
vertices is 2-edge strongly connected. Two vertices $u$ and $v$ are \emph{2-vertex
strongly connected} in $G$ if they remain strongly connected after the removal
of any single vertex except $u$ and $v$ from $G$.
A \emph{2-vertex strongly connected component} (\vscc) of~$G$ is a maximal 
subgraph of~$G$ such that in the subgraph every pair of distinct
vertices is 2-vertex strongly connected.
Edge and vertex connectivity are central properties of graphs and have
many applications~\cite{BangJensenG09,NagamochiI08}, for example in the 
construction of reliable communication networks~\cite{BondyM76} and in the 
analysis of the structure of networks~\cite{Newman10}.

\paragraph*{Our Results.}
In this work we present algorithms that compute the \escc{s} and the \vscc{s}
of a directed graph in $O(n^2)$ time. For \escc{s} we additionally 
provide an algorithm that runs in $O(m^2 / \log n)$ time, which is faster than 
$O(n^2)$ if $m = O(n)$.
Thus we significantly improve upon the previous $ O (m n) $-time algorithms for 
both \escc{s}~\cite{NagamochiW93,GeorgiadisILP15} and \vscc{s}~\cite{Jaberi14}.
For \escc{s} the previous upper bound stood for 20 years.
Our approach immediately generalizes to computing the \emph{$k$-edge strongly connected components} (\kescc{s}) and the \emph{$k$-vertex strongly connected components} (\kvscc{s}).
We give algorithms that, for any integral constant~$ k > 2$, compute (1) the \kescc{s} in time $ O (n^2 \log{n}) $ (improving upon the previous upper bound of $ O(m n) $~\cite{NagamochiW93}) and (2) the \kvscc{s} in time $ O (n^3) $ (improving upon 
the previous upper bound of 
$ O(m n^2) $~\cite{Makino86}).

\paragraph*{Related Work.}
The 2-edge and 2-vertex connected components of an \emph{undirected} graph can 
be determined in linear time~\cite{Tarjan72,HopcroftT73}. In \emph{directed} 
graphs several related problems can be solved in linear time:
Testing whether a graph is 2-edge or 2-vertex strongly 
connected~\cite{Tarjan76,GabowT85,Georgiadis10}, finding all \emph{strong bridges} and 
\emph{strong articulation points}~\cite{ItalianoLS12}, and determining the \emph{2-edge} and
\emph{2-vertex strongly connected blocks}~\cite{GeorgiadisILP15,GeorgiadisILP15b}.
An edge is a strong bridge and a vertex is a strong articulation point, respectively, 
if its removal from the graph increases the number of strongly connected components 
(\scc{s}) of the graph.
Note the difference between `blocks' and `components' in directed graphs: In 
a 2-edge strongly connected block every pair of distinct vertices is 2-edge 
strongly connected; however, as opposed to a \escc, the paths to connect the 
vertices in a block might use
vertices that are \emph{not} in the same block. Each \escc is completely contained in
one 2-edge strongly connected block, i.e., the \escc{s} refine the 2-edge 
strongly connected blocks. In 
\infull{Appendix~\ref{sec:2escbexample}}\inshort{the full version of this paper} we provide a
construction that shows that knowing the 2-edge strongly connected blocks of a graph
does not help in finding its \escc{s}. The relation between blocks and components
for vertex connectivity is analogous.

Georgiadis~et~al.~\cite{GeorgiadisILP15}
and Jaberi~\cite{Jaberi14} described simple algorithms to compute the \escc{s} 
and \vscc{s} in $O(m n)$-time, respectively, and posed as an open problem whether this 
can be improved to linear time as well. 
An $ O (m n) $ running time for computing the \escc{s} was already achieved by 
Nagamochi and Watanabe in 1993~\cite{NagamochiW93}, which in fact solved the more 
general problem of computing the \kescc{s}.
To the best of our knowledge, the fastest known algorithm for computing the \kvscc{s} 
is by Makino~\cite{Makino86} and has a running time of $O(m n^2)$ (when combined
with an $O(m n)$-time algorithm for finding minimum vertex-separators~\cite{Even75,Galil80,HenzingerRG00,Gabow06};
combined with \cite{Georgiadis10} and~\cite{BuchsbaumGKRTW08} it also gives an 
$O(m n)$-time algorithm for \vscc{s}). 
In \emph{undirected} graphs there are linear-time algorithms for computing both 
the 3-edge~\cite{GalilI91} and the 3-vertex~\cite{HopcroftT73} connected components. 
The $k$-edge connected components of an undirected graph
can be computed in time $O(n^2)$~\cite{NagamochiW93}. 
The runtime of Makino's algorithm can for $k$-vertex connected
components in undirected graphs be reduced to $O(n^3)$ by a preprocessing
step~\cite{NagamochiI92}. Thus, to the best of our knowledge, our algorithms
for \kescc{s} and \kvscc{s} match the runtimes for undirected graphs for $k > 3$
(up to a logarithmic factor).

\paragraph*{Techniques.}
We use a \emph{hierarchical graph sparsification} that was introduced by 
Henzinger~et~al.~\cite{HenzingerKW99} for undirected graphs and extended to directed 
graphs and game graphs in~\cite{ChatterjeeH14,ChatterjeeHL15}.
Roughly speaking, this sparsification technique allows us to replace the `$ m $' 
in the $ O(m n) $ running time by an `$ n $', yielding $ O (n^2) $.
Our main technical contribution is to find structural properties of 
connectivity in directed graphs that allow us to apply this technique. 
Note that while various ways of sparsification are used in algorithms for 
\emph{undirected} graphs, such approaches are rarely found for \emph{directed} graphs.
\infull{
For example, Georgiadis~et~al.~\cite{GeorgiadisILP15} showed that a sparse 
certificate for 2-edge strong connectivity can be constructed in $O(m)$ time; the 
sparsification framework by Eppstein~et~al.~\cite{EppsteinGIN97} would yield a
dynamic algorithm with $O(n)$ update time if a sparse \emph{strong} certificate
could be constructed in $O(m)$ time; however, Khanna~et~al.~\cite{KhannaMW98}
showed that even a strong certificate for strong connectivity has size 
$\Omega(m)$, i.e., is not sparse.
We therefore strongly believe that understanding how to apply sparsification for 
directed graphs is interesting in its own right.
}

We briefly present the main ideas behind our algorithm for 2-vertex connectivity. 
The approach for edge connectivity is similar.
The fastest known asymptotic running time of $ O (m n) $ for computing \vscc{s}
can be achieved with the following approach:
Assume that the graph is strongly connected. First find a strong articulation 
point of the graph, i.e., a vertex whose removal increases the number of 
\scc{s}.
Then remove the strong articulation point and compute the \scc{s}. 
For each \scc, recurse
on the subgraph it induces together with the strong articulation point. 
The recursion stops when no strong articulation point is found anymore.
The \scc{s} remaining in the end are the \vscc{s}.
We now explain in which way our algorithm deviates from this scheme.

Let for 2-vertex connectivity a \emph{$2$-isolated set}~$S$
be a set of vertices that (a) 
cannot be reached by the vertices of $V \setminus S$ or (b) that can be reached 
from $V \setminus S$ only through one vertex~$v$. We show that every
\vscc of $G$ contains either only vertices of $S \cup \set{v}$ or only vertices of 
$V \setminus S$. Thus the algorithm can recurse on the subgraphs induced 
by $S \cup \set{v}$ and $V \setminus S$, respectively.
The difference to the straightforward approach is thus the following:
Instead of repeatedly identifying strong articulation points, we focus on 
separating 2-isolated sets of vertices. To see why this is useful, note that the 
incoming edges of the vertices of a 2-isolated set~$S$
consist of the incoming vertices from other vertices of~$S$ and edges
from at most one vertex of $V \setminus S$ to $S$. Thus the number of incoming 
edges of each vertex in~$S$ is bounded by the number of vertices in~$S$. 
We use this insight as follows: When searching for a 2-isolated set, we start 
the search in a subgraph of~$G$ that includes all vertices but only the first 
incoming edge of each vertex. If no 2-isolated set is found, we repeatedly 
double the number of incoming edges per vertex in the subgraph 
until the search is successful. In this way the search will take time $O(n)$ per
vertex in the 2-isolated set. This will allow us to bound the total runtime by $O(n^2)$.
Note that to achieve this running time we cannot afford to compute all \scc{s} 
in each recursive call because the recursion depth might be $\Theta(n)$; we therefore do 
not assume that the input graph is strongly connected.

To correctly identify 2-isolated sets by a search in a proper subgraph of~$G$, 
the algorithm finds \emph{vertex-dominators} in slightly modified 
\emph{flow graphs}. A flow graph is a directed graph with a designated
root where all vertices are reachable from the root. A vertex is a vertex-dominator
in a flow graph if some other vertex can be reached from the root only through this 
vertex. Our algorithms use the linear-time algorithms for finding
dominators~\cite{Tarjan76,GabowT85,BuchsbaumGKRTW08} and 
\scc{s}~\cite{Tarjan72} as subroutines.

In the $O(m^2 / \log n)$-algorithm for \escc{s} we search for 2-(edge-)isolated sets in
subgraphs that are obtained by local breadth-first searches from vertices that lost 
edges in the previous iteration of the algorithm. Such local breadth-first 
searches were first used for B\"uchi games by 
Chatterjee~et~al.~\cite{ChatterjeeJH03}.

\paragraph*{Outline.}
\infull{In the main part of the paper we describe the result for \vscc{s} in
more detail.} In Section~\ref{sec:prelim} the main definitions and the
notation are introduced. In Section~\ref{sec:subgraph} we show when and 
how we can identify a 2-isolated set in a proper subgraph of~$G$. In
Section~\ref{sec:algo} we present the $O(n^2)$-algorithm for \vscc{s}.
In Section~\ref{sec:kscc_short} we outline how the results from
Sections~\ref{sec:subgraph} and~\ref{sec:algo} extend to \kescc{s} and \kvscc{s}.
\inshort{The proofs are given in the full version of this paper.}
\infull{In Appendix~\ref{sec:kscc} we provide a full version with all proofs for \kescc{s} and \kvscc{s}.
In Appendix~\ref{sec:local} we present the $O(m^2 / \log n)$-algorithm for \escc{s}.}

\section{Preliminaries}\label{sec:prelim}
Let $G=(V,E)$ be a directed graph with $m = |E|$ edges and $n = |V|$ vertices.
Except when mentioned explicitly, we only consider simple graphs, i.e., 
graphs without parallel edges.
The reverse graph $\revG{G}$ of $G$ is equal to $(V,\rev{E})$ 
where $\rev{E}$ is the set containing for each edge $ (u, v) \in E $ 
its reverse $ (v, u) $. We use $S \subseteq V$ to denote a subset $S$ of $V$
and $S \subsetneq V$ to denote a proper subset $S$ of $V$.
For any set $S \subseteq V$ we denote by $G[S]$ the subgraph of 
$G$ induced by the vertices in $S$, i.e., the graph $(S, E \cap (S \times S))$.
We call edges from some $u \in V \setminus S$ to some $v \in S$ the \emph{incoming 
edges} of $S$.
The incoming edges of a vertex $v$ in $G$ are denoted by $\In_G(v)$,
the number of incoming edges by $\InDeg_G(v)$; analogously we use 
$\Out_G(v)$ and $\OutDeg_G(v)$ for outgoing edges.
We denote by $G \setminus V'$ the graph 
$G[V \setminus V']$ and by $G \setminus E'$ the graph 
$(V,E \setminus E')$ for an arbitrary set of vertices $V' \subseteq V$ and 
an arbitrary set of edges $E' \subseteq E$.

\paragraph*{Strong Connectivity.} 
A subgraph $G[S]$ induced by some set of vertices $S$ is \emph{strongly connected} 
if for every pair of distinct vertices $u$ and $v$ in $S$ there exists a path 
from $u$ to $v$ and a path 
from $v$ to $u$ in $G[S]$. A 
single vertex is considered 
strongly connected. The \emph{strongly connected components} (\scc{s}) 
of $G$ are its maximal strongly connected subgraphs and form a 
partition of $V$. A strongly connected subgraph with no outgoing edges is a 
\emph{bottom \scc} (\bscc), a strongly connected subgraph with no incoming edges 
is a \emph{top \scc} (\tscc). By definition, \bscc{s} and \tscc{s} are maximal. 
Every graph $G$ contains at least one \bscc and at least one \tscc. If $G$ is 
not strongly connected, then there exist both
a \bscc and a \tscc that are disjoint and thus one of them contains at
most half of the vertices of $G$. Note that a \bscc in $G$ is a \tscc
in $\revG{G}$ and vice versa. We further use that when a set of vertices $S$
cannot be reached by any vertex of $V \setminus S$ in $G$, then $G[S]$ contains
a \tscc of $G$. 
\infull{
The \scc{s} of a graph can be computed in $O(m)$~time~\cite{Tarjan72}.
} 

\paragraph*{Strong 2-Vertex Connectivity.}
A vertex $v \in V$ is a \emph{strong articulation point} if the removal of $v$ from 
$G$ increases the number of \scc{s} in $G$. 
\infull{
All strong articulation points of a graph 
can be found in time $O(m)$~\cite{ItalianoLS12}.
} 
Two (simple) paths are 
\emph{internally vertex-disjoint} if they do not share a vertex except possibly their 
endpoints. Two distinct vertices $u$ and $v$ are \emph{2-vertex strongly connected} 
in~$G$ if they are strongly connected and remain strongly connected after the
removal of any vertex except $u$ and $v$ from~$G$. 
If there is no edge between $u$ and $v$, then it holds that $u$ and $v$ are 
2-vertex strongly connected if and only if there exists two internally 
vertex-disjoint paths from~$u$ to~$v$ and two internally vertex-disjoint paths from~$v$ to~$u$~\cite{GeorgiadisILP15}.
A subgraph $G[S]$ induced by some set of vertices $S$ is 
\emph{2-vertex strongly connected} if every pair of distinct 
vertices~$u$ and~$v$ in~$S$ is 2-vertex strongly connected in~$G[S]$.
The \emph{2-vertex strongly connected components}\footnote{
Our definitions follow~\cite{GeorgiadisILP15}, while~\cite{Jaberi14,GeorgiadisILP15b} 
use slightly different definitions. The \vscc{s} of \cite{Jaberi14,GeorgiadisILP15b} 
can be determined in $O(n)$ time from the \vscc{s} defined here.}
(\vscc{s}) of a graph are its maximal 2-vertex strongly connected subgraphs. 
Equivalently, the \vscc{s} are
the maximal strongly connected subgraphs such that none of the subgraphs 
contains a strong articulation point. This definition of \vscc{s} allows for 
\emph{degenerate} \vscc{s} with less than three vertices.
While the \escc{s} form a partition of the vertices of the graph, the \vscc{s} 
form a partition of a subset of the edges.
In the remainder of the paper we omit ``strong(ly)'' from the above definitions
whenever it is clear from the context.

\paragraph*{Flow Graphs.}
We define the \emph{flow graph} $G(r)$ to be the graph $G$ with a vertex $r \in V$ 
designated as the root and with all vertices not reachable from $r$ removed.
A \emph{vertex-dominator} in~$G(r)$ is a vertex $v \in V \setminus \{r\}$ for
which there exists a vertex $u \in V \setminus \{r,v\}$ such that
$u$ is reachable from~$r$ and every path from~$r$ to~$u$ contains~$v$.
We say that $v$ \emph{dominates} $u$ in $G(r)$.
Note that in contrast to articulation points the removal of a vertex-dominator 
from $G$ might not increase the number of \scc{s} but instead might remove edges 
between \scc{s}.
\infull{
The vertex-dominators of a flow graph can be computed in 
linear time~\cite{BuchsbaumGKRTW08}.
}
\section{New top \scc{s} and dominators in subgraphs}\label{sec:subgraph}
Let an \emph{isolated set~$S$ w.r.t.\ 2-vertex-connectivity} (\emph{2-isolated set})
be a set of vertices with (1) incoming edges from at 
most one vertex and for which (2) there exist vertices without edges to $S$ in $G$.
2-isolated sets can be used to design a divide-and-conquer based algorithm
for the following reason: Let $T$ be the vertex set of a \vscc. The \vscc $G[T]$ 
is (1) strongly connected 
and (2) for any proper subset~$S$ of $T$ such that there exists a set of
vertices~$U$ in $T$ that has no edge to any vertex of~$S$, there are at least 
two vertices in $T \setminus (S \cup U)$ that connect $U$ with the 
vertices in~$S$. 
\begin{figure}
\centering
\begin{tikzpicture}
\tikzstyle{graph}=[circle,draw,minimum size=10mm]
\tikzstyle{vertex}=[circle,draw,fill,inner sep=1.2pt, solid]
\tikzstyle{arrow}=[->,line width=0.5pt,>=stealth',thick]
\node[graph] (S) at (-1.5,0) {$S$};
\node[graph] (U) at (1.5,0) {$U$};
\node[vertex] (v) at (0,0.1) {};
\node[vertex] (u) at (0,-0.2) {};

\path (U) edge[arrow] (v) edge[arrow] (u)
(v) edge[arrow] (S)
(u) edge[arrow] (S)
(S) edge[arrow, dotted, bend angle=20, bend left] (U)
(S) edge[arrow, dotted, bend angle=30, bend left] (U)
(S) edge[arrow, dotted, bend angle=25, bend right] (U)
;
\end{tikzpicture}
\end{figure}
Thus if we detect a set of vertices~$S$ that (a) 
cannot be reached by the vertices of $V \setminus S$ or (b) that can be reached 
from $V \setminus S$ only through one vertex~$v$, then we know that each \vscc 
of~$G$ contains either only vertices of $S \cup \set{v}$ or only vertices of 
$V \setminus S$. A 2-isolated set satisfies (a) or (b). 
Our algorithm repeatedly identifies specific 2-isolated sets~$S$ and recurses 
on the subgraphs induced by $S \cup \set{v}$ and $V \setminus S$, respectively.
As the recursion depth can be $\Theta(n)$, to achieve an $o(mn)$
running time, we cannot afford to look at all edges in each level of recursion.
Thus our algorithms are based on the following question: \emph{Can we identify
2-isolated sets by searching in a proper subgraph of $G$?}
Note that whenever an articulation point $v$ is removed from a strongly 
connected graph~$G$, then there exist both a \tscc and a \bscc in 
$G \setminus\set{v}$ that were adjacent to $v$ in $G$ and are disjoint.
Let $T$ be the vertices in the \tscc in $G \setminus\set{v}$. Observe that $T$ 
is a 2-isolated set in~$G$. Further, if $T$ contains only a few vertices,
then each vertex in $T$ has a low in-degree in~$G$ because all incoming 
edges to vertices in~$T$ in~$G$ come from~$v$ or other vertices of $T$. 
In our algorithm we search for such ``almost 
\tscc{s}'' $G[T]$ in the subgraph of $G$ induced by vertices with low in-degree,
which only takes time linear in the number of edges in this subgraph.
We do the same on $\revG{G}$ to detect small almost \bscc{s}.
\begin{definition}
	A set of vertices~$T$ induces an \emph{almost \tscc} in 
	$G$ with respect to a vertex~$v$ if $G[T]$ is a 
	\tscc in $G \setminus\set{v}$ but has incoming edges from $v$ in~$G$.
\end{definition}
Given a vertex $v$ such that an almost \tscc induced by $T$ w.r.t.\ $v$ 
exists, the top \scc $G[T]$ can be identified in a subgraph of $G\setminus \set{v}$
in time linear in the number of edges in the subgraph as long as 
it is contained in the subgraph. But how can we 
identify the vertex~$v$ without looking at 
the whole graph? Assume there exists a vertex $r \ne v$ that is not in $T$ but 
can reach $v$. Since $G[T]$ is a \tscc in $G\setminus \set{v}$, 
it follows that $v$ dominates every vertex of $T$ in the flow graph $G(r)$. 
This still holds
in any subgraph of $G$ as long as $r$ can reach $T$ in the subgraph. If additionally
all incoming edges of the vertices in $T$ are present in the subgraph,
we can identify~$v$ and~$T$ in time linear in the number of edges in the subgraph by 
finding the vertex-dominator~$v$ in the flow graph with root~$r$ and the \tscc $G[T]$ 
in the subgraph with~$v$~removed. Thus, instead of finding the right~$v$, we only
have to find the right~$r$. As edges are missing in the subgraph, it is not a-priori
clear how to choose~$r$, but, as shown below, we can use an artificial vertex as root~$r$.
Hence our approach is to first search for vertex-dominators~$v$ in a subgraph 
with an additional artificial root and then for a \tscc in the subgraph with $v$
removed. When the search is successful, we recurse separately on the almost \tscc 
and the remaining graph.

In our algorithm we cannot afford to identify all \scc{s} in the current 
graph~$G$ as we only want to spend time proportional to the edges in a proper subgraph
of~$G$; thus we cannot assume that the graph we are considering is strongly 
connected. This means that, in contrast to strongly connected 
graphs~\cite{ItalianoLS12}, when we 
identify a vertex-dominator~$v$ in~$G(r)$, the vertex~$v$ might not necessarily 
be an articulation point in~$G$. However, for an almost \tscc w.r.t.\ $v$ we still 
know that the set of vertices~$T$ in the almost \tscc is a 2-isolated set, i.e., all 
vertices of $V \setminus (T \cup \set{v})$ that can 
reach $T$ in $G$ can reach $T$ only through $v$.
Thus there cannot be two internally vertex-disjoint paths 
from any vertex of $G \setminus (T \cup \set{v})$ to any vertex of $T$.
This intuition about almost \tscc{s} is summarized in the following lemma, which
we use to show the correctness of our approach.
\begin{lemma}\label{lem:newtscc}
Let $v$ be a vertex such that some set of vertices~$T$ induces an almost \tscc with 
respect to $v$ in $G$.
Let $W = V \setminus (T \cup \set{v})$. If $W \ne \emptyset$, then there do not 
exist two internally vertex-disjoint paths from any vertex of~$W$ to any 
vertex of~$T$ in $ G $,
i.e., no vertex of~$W$ is 2-vertex-connected to any vertex of~$T$.
Additionally, the vertex~$v$ is a vertex-dominator in~$G(r)$ for every~$r \in W$
that can reach~$v$ in~$G$.
\end{lemma}

Let $\subG{G}{h} = (\subV{h},\subE{h})$ be a subgraph of a directed graph 
$G = (V,E)$, i.e., $\subV{h} \subseteq V$ and $\subE{h}
\subseteq G[\subV{h}]$. We use the index $h$ to identify specific subgraphs.
In the remainder of this section we want to characterize which almost \tscc{s} in $G$ we 
can identify in $\subG{G}{h}$. Let $v$ be a 
vertex such that an almost \tscc w.r.t.\ $v$ exists in~$G$. To identify $v$ as a
vertex-dominator in a flow graph, we define below
a graph created from $\subG{G}{h}$ with an auxiliary root.
Let the \emph{white} vertices $\good{G}{h} \subseteq \subV{G}$ be the set of vertices
for which we have the guarantee that for each vertex in $\good{G}{h}$ its incoming
edges in $\subG{G}{h}$ are the same as in $G$. Let $\bad{G}{h} = \subV{h} 
\setminus \good{G}{h}$ be the \emph{blue} vertices, which might miss incoming
edges in $\subG{G}{h}$ compared to $G$.
We show that as long as the vertices in the 
almost \tscc are white, i.e., are not missing incoming edges 
in~$\subG{G}{h}$, an almost \tscc w.r.t.\ 
a vertex~$v$ in $\subG{G}{h}$ is an almost \tscc w.r.t.\ $v$ in $G$ and vice versa. 
In contrast, no conclusions can be drawn from 
an almost \tscc in $\subG{G}{h}$ that includes blue vertices.
\begin{definition}\label{def:subgraph}
For a given subgraph $\subG{G}{h} = (\subV{h},\subE{h})$ of a directed graph 
$G = (V,E)$ and a set of blue vertices $\bad{G}{h}$ that contains all vertices 
that have fewer incoming edges in $\subG{G}{h}$ than in $G$, we define the flow 
graph $\subGv{G}{h}(\rootv{G}{h})$ as follows.
If $\lvert \bad{G}{h} \rvert \ge 2$, let $\subGv{G}{h}$ be the graph 
$\subG{G}{h}$ with an additional vertex~$\rootv{G}{h}$ and an additional edge 
from $\rootv{G}{h}$ to each vertex in~$\bad{G}{h}$. If $\bad{G}{h}$ contains a 
single vertex, we name it $\rootv{G}{h}$ and let $\subGv{G}{h} = \subG{G}{h}$.
\end{definition}
\infull{Note that if we would define 
the flow graph $\subGv{G}{h}(\rootv{G}{h})$ in the case $\lvert \bad{G}{h} 
\rvert = 1$ in the same way as for $\lvert \bad{G}{h} \rvert > 1$, then in 
$\subGv{G}{h}(\rootv{G}{h})$
the root $\rootv{G}{h}$ could reach the vertices in $\good{G}{h}$ only through
the vertex in $\bad{G}{h}$, i.e., the vertex in $\bad{G}{h}$ would be 
a vertex-dominator in $\subGv{G}{h}(\rootv{G}{h})$ independent of the underlying
graph $G$.}
In the following consider a subgraph $\subG{G}{h}$ and a set of vertices 
$\subV{h}$ partitioned into $\bad{G}{h}$ and $\good{G}{h}$ as\infull{ defined} above;
the statements for $\subGv{G}{h}$ hold whenever $\subGv{G}{h}$ is defined.
\begin{lemma}\label{lem:tscc}
A set of white vertices $T \subseteq \good{G}{h}$ induces a \tscc in 
$\subG{G}{h}$ and $\subGv{G}{h}$, respectively, if and only if it induces a \tscc in $G$.
\end{lemma}
If white vertices $T$ induce an almost \tscc $G[T]$ with respect to $v$, all 
incoming edges, and thus $v$,
are present in $\subG{G}{h}$. 
This implies the following corollary.
\begin{corollary}\label{cor:corr}
A set of white vertices $T \subseteq \good{G}{h}$ induces an almost \tscc 
with respect to a vertex $v\in V$ in~$\subG{G}{h}$ and 
$\subGv{G}{h}$, respectively, if and 
only if it induces an almost \tscc with respect to $v$ in~$G$.
\end{corollary}


The following lemma specifies which almost \tscc{s} w.r.t.\ a vertex~$v$ in $G$ we 
can identify by searching for vertex-dominators in $\subGv{G}{h}(\rootv{G}{h})$ 
based on the reachability of $v$ from the vertices in $\bad{G}{h}$.
\begin{lemma}\label{lem:finddom}
Assume $\bad{G}{h} \ne \emptyset$, let $T \subseteq \good{G}{h}$ be a set 
of white vertices, and let $v \in V$ be such that there exists an almost \tscc $G[T]$ 
with respect to $v$ in $G$.
If $v$ is either not in $\bad{G}{h}$ and can be reached from a vertex of 
$\bad{G}{h}$ or $v$ is in $\bad{G}{h}$ and $\lvert \bad{G}{h}\rvert \ge 2$,
then $v$ is a dominator in $\subGv{G}{h}(\rootv{G}{h})$.
\end{lemma}

In the following section we define specific subgraphs $\subG{G}{h}$
that allow us to identify an almost \tscc in $G$ that has at most a 
certain size by searching for vertex-dominators~$v$ in $\subGv{G}{h}(\rootv{G}{h})$
and \tscc{s} in $\subG{G}{h} \setminus \set{v}$.
We additionally have to consider one special case, namely if $v$ is the only 
vertex in $\bad{G}{h}$ and an almost \tscc w.r.t.\ $v$ exists. In this case 
we have $\rootf{G}{h} = v$.
We explicitly identify almost \tscc{s} with respect to this vertex. 

\section{\vscc{s} in $O(n^2)$ time}\label{sec:algo}
In this section we provide some intuition for the algorithm and outline 
its analysis.
\infull{All proofs are given in Appendix~\ref{sec:kscc}.}
To find vertex-dominators, articulation points, and \scc{s} the known linear 
time algorithms are used (see Section~\ref{sec:intro}).

Let $G = (V,E)$ be a simple directed graph. We consider for $i \in \mathbb{N}$ 
the subgraphs $G_i = (V, E_i)$ of $G$ where~$E_i$ contains for 
each vertex of $V$ its first $2^i$ incoming edges in $E$ (for some arbitrary but
fixed ordering of the incoming edges of each vertex). Note that when
$i \ge \log (\max_{v \in V}{\InDeg_G(v)})$, then $G_i = G$.
Let $\mmdeg$ be the minimum of $\max_{v \in V}{\InDeg_G(v)}$ and 
$\max_{v \in V}{\OutDeg_G(v)}$. Following Definition~\ref{def:subgraph}, 
the set~$\bad{G}{i}$ contains all vertices with in-degree more 
than $2^i$ in~$G$. %

\begin{procedure}
\caption{2vSCC($G$)} 
\label{alg1}
\SetKwProg{myproc}{}{}{}
	\For{$i \leftarrow 1$ \KwTo $\lceil \log \mmdeg \rceil - 1$}{
			$(S, Z) \leftarrow \levelsearch(G, i)$ \;
			\tcc*[l]{$Z$ contains $v$ if $G[S]$ is \emph{almost} top or bottom \scc w.r.t.\ $v$}
			\If{$S \ne \emptyset$}{
				\Return{$\rec(G[S \cup Z]) \cup \rec(G[V \setminus S])$}\;
			}
	}
	$(S, Z) \leftarrow \allsearch(G)$\;
 	\tcc*[l]{$Z$ contains $v$ if $G[S]$ is \emph{almost} top \scc w.r.t.\ $v$}
	\eIf{$S \ne \emptyset$}{
		\Return{$\rec(G[S \cup Z]) \cup \rec(G[V \setminus S])$}\;
	}{
	\Return{$\set{G}$}\;
	}
\end{procedure}

Let $S$ be a set of at most $2^i$ vertices that induces a strongly connected 
subgraph~$G[S]$ of $G$ such that $G[S]$ is a top \scc or an 
almost top \scc with respect to some vertex $v$. Since the only
edges from vertices of $V \setminus S$ to $S$ are from $v$, the in-degree
of each vertex in $S$ can be at most $2^i$.
By applying the results from the previous section, we 
show that we can detect such a set $S$ by searching for \scc{s}
and vertex-dominators in the graphs~$\subGv{G}{i}$ constructed from $G_i$ 
with the artificial root $\rootv{G}{i}$ as in Definition~\ref{def:subgraph}. 
\begin{lemma}\label{lem:level}
If a set of vertices $S$ with $\lvert S \rvert \le 2^i$ 
induces a \tscc or an almost \tscc in $G$ with respect to 
some vertex~$v$, then $S \subseteq V\setminus \bad{G}{i}$.
\end{lemma}
To find \bscc{s} and almost \bscc{s} we also search for top 
\scc{s} in $\revG{G}$. The search for both top and bottom \scc{s} 
ensures that whenever an (almost) \tscc and a disjoint (almost) \bscc exist in~$G$, 
we only spend time proportional to the smaller one.
This search is performed in Procedure~\ref{alglevel}, which fulfills the 
following guarantee.
\begin{lemma}\label{lem:find}
	If for some integer $1 \le i < \log \mmdeg$ and 
	$\G \in \set{G, \revG{G}}$ there exists a set of vertices $T \subseteq 
	V \setminus \bad{\G}{i}$ that induces in $\G$ a \tscc or an almost \tscc with 
	respect to some vertex~$v$ with $T \subsetneq 
	V \setminus \set{v}$, then $\levelsearch(G, i)$ returns a non-empty set $S$.
\end{lemma}

In Procedure~\ref{alg1} we start the search for (almost) top \scc{s} at $i = 1$. 
Whenever the search is not successful, we increase~$i$ by one, until we have 
$G_i = G$ or $\revG{G}_i = \revG{G}$. 
For the search the Procedure~\ref{alglevel} is used
as long as $2^i < \mmdeg$,
i.e., both $\bad{G}{i}$ and $\bad{\revG{G}}{i}$ are non-empty, and the 
Procedure~\ref{algsearch} afterwards. Procedure~\ref{algsearch} identifies
an (almost) top \scc in $G$ if one exists by using the known procedures 
for finding \scc{s} and articulation points.
In this way we can show that whenever we had to go up to~$i^*$ or 
had to use Procedure~\ref{algsearch} to identify an 
(almost) top or bottom \scc in $G$, the identified subgraph contains 
$\Omega(2^{i^*})$ vertices, where $i^* = \lceil \log \mmdeg \rceil$ for Procedure~\ref{algsearch}.
This will imply that the search in $G_i$ and $\revG{G}_i$ for $i$ up to $i^*$
takes time $O(n \cdot 2^{i^*})$ which is $O(n \cdot \min\{\lvert S \rvert, \lvert 
V\setminus S \rvert\})$. This will allow us to bound the total running time by $O(n^2)$.

\begin{procedure}
\caption{2IsolatedSetLevel($G$, $i$)} 
\label{alglevel}
\ForEach{$\G \in \set{G, \revG{G}}$}{
	\tcc{$2^i < \max_{v \in V}{\InDeg_\G(v)} \Longrightarrow \bad{\G}{i} \ne \emptyset$}
	construct $\G_i = (V, E_i)$ with 
	$E_i = \cup_{v\in V}\{\text{first } 2^i \text{ edges in } \In_\G(v)\}$\;
	$\bad{\G}{i} = \{v \mid \InDeg_\G(v) > 2^i\}$\;
	$S \leftarrow \topsccex(\G_i, \bad{\G}{i})$\;
	\If{$S \ne \emptyset$}{
		\Return{$(S, \emptyset)$}\;
	}
	construct flow graph $\subGv{\G}{i}(\rootv{\G}{i})$
	\tcc*[f]{see Definition~\ref{def:subgraph}}\;
	\If{exists vertex-dominator $v$ in $\subGv{\G}{i}(\rootv{\G}{i})$}{
		$S \leftarrow \topsccex(\G_i\setminus \set{v}, \bad{\G}{i})$\;
		\Return{$(S, \set{v})$}
	}
	\ElseIf{$\lvert \bad{\G}{i} \rvert = 1$ and 
	$\exists$ \tscc $\subsetneq V \setminus \set{\rootv{\G}{i}}$ 
	in $\G_i \setminus \set{\rootv{\G}{i}}$}{
		$S \leftarrow \topscc(\G_i\setminus \set{\rootv{\G}{i}})$\;
		\Return{$(S, \set{\rootv{\G}{i}})$}
	}
}
\Return{$(\emptyset, \emptyset)$}\;
\end{procedure}

Let $\G_i \in \set{G_i, \revG{G}_i}$. The Procedure~\ref{alglevel} first 
searches for a \tscc in $\G_i$ that does not contain a vertex of 
$\bad{\G}{i}$. If no such 
\tscc is found, the flow graph $\subGv{\G}{i}(\rootv{\G}{i})$ 
is constructed and searched for vertex-dominators. If a vertex-dominator~$v$ is 
found, a \tscc in~$\G_i \setminus \set{v}$ that does not contain a vertex of 
$\bad{\G}{i}$ is found; one can show that such a \tscc always exists.
We additionally have to consider the special case when 
$\lvert \bad{\G}{i} \rvert = 1$. In this case we have $\bad{\G}{i} = 
\set{\rootv{\G}{i}}$ and we want to detect when there exists an
almost \tscc $\G[T]$ induced by some set of vertices $T$
with respect to $\rootv{\G}{i}$ in $\G_i$
such that $V \setminus (T \cup \set{\rootv{\G}{i}})$ is not empty.
We use Procedure~$\topsccex(H, B)$ to denote the search for a \tscc induced 
by vertices $S$
in a graph~$H$ such that $S$ does not contain a vertex of~$B$. Such a \tscc can
simply be found by marking \tscc{s} in a standard \scc algorithm.
We let all procedures that search for an \scc return the set of vertices~$S$ 
in the \scc instead of the subgraph~$G[S]$.

\begin{procedure}
\caption{2IsolatedSet($G$)} 
\label{algsearch}
	$S \leftarrow \topscc(G)$\;
	\If{$S \subsetneq V$}{
		\Return{$(S, \emptyset)$}\;
	}
	\If{exists articulation point $v$ in $G$}{
		$S \leftarrow \topscc(G \setminus \set{v})$\;
		\Return{$(S, \set{v})$}\;
	}
	\Return{$(\emptyset, \emptyset)$}\;
\end{procedure}

If no call to Procedure~\ref{alglevel} could identify an (almost) top or 
bottom \scc, we check
in Procedure~\ref{algsearch} whether the graph is strongly connected
and either make progress by separating strongly connected components
from each other or by finding an articulation point in the strongly
connected graph. If an articulation point $v$ is found, disjoint top and 
bottom \scc{s} exist after the removal of the articulation point~$v$. 
Procedure~\ref{algsearch} returns a top \scc in $G \setminus \set{v}$ 
in this case. If the graph~$G$ is strongly connected and does not contain an
articulation point, then $G$ is a \vscc. In this case the 
Procedure~\ref{algsearch} returns the empty set, the recursion stops, and 
$\rec(G)$ returns $G$.

Whenever the algorithm identifies an (almost) top or bottom \scc induced by a 
set of vertices $S$,
it recursively
calls itself on $G[S \cup Z]$ and $G[V \setminus S]$ for $Z = \emptyset$ or 
$Z = \set{v}$, respectively. We use Lemma~\ref{lem:newtscc} to show
that in this case every \vscc of $G$ is completely contained in either
$G[S \cup Z]$ or $G[V \setminus S]$, which will imply the correctness of the algorithm.
\begin{theorem}[Correctness]\label{th:corr}
Let $G$ be a simple directed graph. $\rec(G)$ computes the \vscc{s} of $G$.
\end{theorem}

By stopping the recursion when the number of vertices is a small 
constant and distinguishing between the number of vertices $n'$ at the current
level of the recursion and the total number of vertices $n$,
we can show that the runtime of $O(n' \cdot \min\{\lvert S \rvert, \lvert 
V\setminus S \rvert\})$ without recursion leads to a total runtime of $O(n^2)$.

\begin{theorem}[Runtime]\label{th:timeV1}
	Procedure~\ref{alg1} can be implemented in time $O(n^2)$.
\end{theorem}

\section{Extension to \kscc{s}}\label{sec:kscc_short}
For any integral constant~$k > 2$ the presented algorithm extends to computing 
the $k$-edge and the $k$-vertex strongly connected components.
\infull{In this section we outline the necessary changes, the details including all proofs are 
given in Appendix~\ref{sec:kscc}.}
\inshort{In this section we outline the necessary changes, see the full version for
more details.}

Let an \emph{element} of a graph $G$ denote an edge when \kescc{s} are searched for
and a vertex when \kvscc{s} are searched for.
We first extend the concepts of bridges, articulation points, 
and dominators from a single element to sets of elements with size less than $k$.
A \emph{separator w.r.t.\ $k$-connectivity} (\emph{$k$-separator})
is a minimal set of elements such 
that the set contains less than $k$ elements and its removal from the graph 
increases the number of \scc{s} in the graph. Two distinct vertices
$u$ and $v$ are \emph{$k$-\textup{(}strongly-\textup{)}connected} if they are 
strongly connected and they remain strongly connected after the 
removal of any less than $k$ elements different from $u$ and $v$ from $G$. 
The \emph{$k$-strongly connected
components} (\kscc{s}) of a graph $G$ are its maximal subgraphs $G[S]$ such that every
pair of distinct vertices $u$ and $v$ in $S$ is $k$-connected in $G[S]$.%
\infull{For vertex-connectivity this definition
allows for degenerate \kvscc{s} with $k$ or less vertices. 
Given the \kvscc{s}, the degenerate \kvscc{s} can be identified in linear time.}

In a flow graph $G(r)$ a \emph{dominator $Z$ w.r.t.\ $k$-connectivity} 
(\emph{$k$-dominator}) is a \emph{minimal} set of less than 
$k$ elements in $G(r) \setminus \set{r}$ such that
there exists a vertex $u \in G(r) \setminus (\set{r} \cup Z)$
such that $u$ is reachable from $r$ and every path from $r$ to $u$
contains an element of $Z$. \infull{We say that $Z$ \emph{$k$-dominates} $u$ in $G(r)$.
Note that, in contrast to $k$-separators, the removal of a $k$-dominator from $G$ 
might not increase the number of \scc{s} but instead remove edges between \scc{s}.}
A $k$-dominator in a flow graph $G(r)$ and a $k$-separator in a graph $G$ can
for edge-connectivity be found in time $O(m \log n)$~\cite{Gabow95} and for 
vertex-connectivity in time $O(mn)$~\cite{Even75,Galil80,HenzingerRG00,Gabow06}.

A set of vertices $T$ induces an \emph{almost \tscc w.r.t.\ $k$-connectivity} 
(\emph{$k$-almost \tscc}) in $G$ with respect to a set of elements $Z$ with 
$\lvert Z \rvert < k$ if $G[T]$ is a \tscc in $G \setminus Z$ but 
has, for vertex-connectivity, incoming edges from \emph{each of} the vertices in $Z$, 
or, for edge-connectivity, \emph{all} the edges in $Z$ as incoming edges in $G$. 

We adapt our algorithm as follows. For edge-connectivity we use different flow 
graphs: \enum{1} We contract all vertices in 
$\bad{G}{i}$ to a single vertex, while keeping all edges between the vertices in 
$\bad{G}{i}$ and the remaining vertices
as parallel edges. \enum{2} We take the new contracted vertex as the root of the 
flow graph. With these definitions it is rather straightforward to extend the algorithm
to $\kescc{s}$.

The extension to $k > 2$ is more complicated for vertex-connectivity. 
In particular, we have to deal with the case $0 < \lvert \bad{G}{i} \rvert < k$.
Note that in this case we cannot use an additional vertex that we connect to 
the vertices of $\bad{G}{i}$ as the root in the flow graph because the vertices
of $\bad{G}{i}$ would be a $k$-dominator in this flow graph independent of the 
underlying graph $G$. To be able to identify a set $Z \cap \bad{G}{i} \ne \emptyset$
with $\lvert Z \rvert < k$ for which a $k$-almost \tscc exists in $G$, 
we use $\lvert \bad{G}{i} \rvert < k$ different flow graphs.
If the search in the $\lvert \bad{G}{i} \rvert$
flow graphs is not successful, we additionally search for a 
$(k-\lvert \bad{G}{i} \rvert)$-separator in $G_i \setminus \bad{G}{i}$\inshort{.}\infull{
to detect the case $Z \supseteq \bad{G}{i}$.
If such a $(k-\lvert \bad{G}{i} \rvert)$-separator $Z'$ exists, 
then $Z = Z' \cup \bad{G}{i}$ contains less 
than $k$ vertices and there exists a $k$-almost \tscc induced by some set of vertices 
$T \subsetneq V \setminus Z$ in $G_i$.}
These changes give the following result.
\begin{theorem}
For any integral constant $k > 2$ \kescc{s} can be computed in time $O(n^2 \log n)$ 
and \kvscc{s} in time $O(n^3)$. \escc{s} can be computed in time $O(n^2)$.
\end{theorem}

\section*{Acknowledgements}
We would like to thank Giuseppe Italiano for suggesting the problem and 
Slobodan Mitrovi\'c for helpful discussions. V.~L.\ would like to thank 
Christian Tschabuschnig for his help in improving the readability of the algorithms.

\printbibliography[heading=bibintoc] 

\infull{\clearpage
\section*{Appendix}
\appendix
\section{Missing proofs and extension to \kscc{s}}\label{sec:kscc}

In this section we describe the algorithms for $k$-edge and $k$-vertex strongly
connected components and prove their correctness and running time. This in
particular implies the results presented in the main part of the paper.
We use the definitions of Sections~\ref{sec:prelim} and~\ref{sec:kscc_short}, 
otherwise the section is self-contained. 
We analyze most parts simultaneously for edge and vertex connectivity but
explicitly point out important differences.
Note that the text is meant to be read consistently for either edge or 
vertex connectivity, e.g., when a statement about $k$-connectivity is interpreted 
as a statement about $k$-vertex connectivity, then any occurrence of ``element'' 
has to be interpreted as a vertex and not as an edge.

\subsection{New top \scc{s} and $k$-dominators in subgraphs}\label{sec:ksubgraph}
We introduce the notion of an \emph{isolated set} with respect to $k$-connectivity,
\emph{$k$-isolated set} for short, for the informal introduction of our approach. 
Let for edge connectivity a \emph{$k$-isolated set}~$S \subsetneq V$ be a set of 
vertices with less than $k$ incoming edges and let $U = V \setminus S$. 
Recall that we use ``incoming edges 
of a set of vertices $S$'' to denote the edges from $V \setminus S$ to $S$.
For vertex connectivity, let a \emph{$k$-isolated set}~$S$ be a set of 
vertices with (1) incoming edges from less than $k$ vertices and for which (2) 
the set of vertices~$U \subseteq V \setminus S$ that have no edges to vertices 
of $S$ is not empty. 
Clearly, for both edge and vertex connectivity, no element of $S$ is 
$k$-connected to any vertex of $U$ 
because there cannot exist $k$ disjoint paths from any vertex of $U$ to 
any vertex of $S$. 
the graph is strongly connected and does not contain a $k$-separator, i.e., the 
graph is $k$-connected.
Think of the following simple recursive algorithm to output the \kscc{s} of graph~$G$: 
\begin{compactenum}
	\item Find a $k$-isolated set $S$ in $G$.
	\item If none exists, output $G$.
	\item Otherwise recurse on the graphs induced by $S$ and $U$; for 
	vertex connectivity add the vertices in $V \setminus (S \cup U)$ to both $S$ 
	and $U$ before the recursion. 
\end{compactenum}
Our algorithms follow this scheme for specific $k$-isolated sets.
As the recursion depth can be $\Theta(n)$, to achieve an $o(mn)$
running time, we cannot afford to look at all edges in each level of recursion.
Thus our algorithms are based on the following question:
\begin{quote}
\emph{Can we identify $k$-isolated sets by searching in a proper subgraph of $G$?}
\end{quote}
We first explain the $k$-isolated sets our algorithms identify, formalize the 
correctness idea outlined above for this kind of $k$-isolated sets, and provide an 
intuition why and when they can be identified in proper subgraphs of $G$.
We then formalize the latter in the remaining part of this subsection. 
The results in this subsection are formulated for general subgraphs such that 
they can also be used for the $O(m^2 / \log n)$-algorithm for \escc{s} presented 
in Appendix~\ref{sec:local}. Our algorithms for \kscc{s} are described in 
Subsection~\ref{sec:kalgo}.

The $k$-isolated sets we identify are \tscc{s} and $k$-almost \tscc{s}. 
Recall the definition of $k$-almost \tscc{s}: 
A set of vertices $T$ induces a \emph{$k$-almost \tscc} in $G$ with 
respect to a set of elements $Z$ with 
$\lvert Z \rvert < k$ if $G[T]$ is a \tscc in $G \setminus Z$ but 
has, for vertex connectivity, incoming edges from \emph{each of} the vertices in $Z$, 
or, for edge connectivity, \emph{all} the edges in $Z$ as incoming edges in $G$.  
For the sake of a compact formulation, we introduce the notation \nt{T}{Z}{G}
for a $k$-almost \tscc in $G$ with respect to a set of elements $Z$ 
with $\lvert Z \rvert < k$ induced by the set of vertices~$T$.
For edge connectivity each $k$-almost \tscc is induced by a $k$-isolated set. 
For vertex connectivity this holds whenever there exist vertices without edges 
to the $k$-almost \tscc; our algorithm only identifies $k$-almost \tscc{s} that are 
induced by $k$-isolated sets.

For the intuition behind the definition of $k$-almost \tscc{s}, think of a 
strongly connected graph that contains a $k$-separator $Z$ (see 
Section~\ref{sec:kscc_short}). Recall that we 
require $k$-separators to be \emph{minimal} (with respect to set inclusion), 
i.e., $Z$ is a minimal set 
of less than $k$ elements such that $G \setminus Z$ is not strongly connected.
For a $k$-separator~$Z$ there exist both a \tscc and a \bscc in 
$G \setminus Z$ that were adjacent 
to $Z$ in $G$ and are disjoint. Let $T$ be the vertices in a \tscc in 
$G \setminus Z$ (there can be more than one for vertex connectivity). Observe 
that $T$ is a $k$-isolated set in $G$ and induces a \nt{T}{Z}{G}. Further, 
\tscc{s} can be identified in time linear in the number of edges in $G$ by a 
standard \scc algorithm by simply marking the \scc{s} without incoming edges.
To see why the notion of $k$-isolated sets is helpful when searching in a 
subgraph of $G$, note the following: 
If $T$ contains only a few vertices, then each vertex in $T$ has low 
in-degree in~$G$ because all incoming edges of a vertex of~$T$ in~$G$ either come 
from other vertices in~$T$ or, for edge connectivity, are the edges in $Z$, and, 
for vertex connectivity, come from the vertices in~$Z$. 
In our algorithms for \kscc{s} we search for \tscc{s} and $k$-almost 
\tscc{s} in the subgraph of $G$ induced by vertices with low in-degree.
We do the same on $\revG{G}$ to detect small \bscc{s} and small $k$-almost \bscc{s}
(defined analogously).

To identify a \nt{T}{Z}{G}, we do not only have to find the \tscc $G[T]$ in 
$G \setminus Z$ but first have to identify the set of elements $Z$.
Assume there exists a vertex $r \notin Z$ that is not in $T$ but 
can reach all elements in $Z$. Since $G[T]$ is a \tscc in $G\setminus Z$,
it follows that $Z$ $k$-dominates every vertex of $T$ in the flow graph $G(r)$.
We formalize this observation and the intuition about the correctness of our 
approach of repeatedly identifying \tscc{s} and $k$-almost \tscc{s} (that are also 
$k$-isolated sets) in the following lemma. 
\begin{lemma}[Extension of Lemma~\ref{lem:newtscc} to $k$-connectivity]
\label{lem:knewtscc}
Let $T$ and $Z$ be such that a \nt{T}{Z}{G} exists. Let $W = V \setminus T$ 
for edge connectivity and let $W = V \setminus (T \cup Z)$ for 
vertex connectivity. Assume $W \ne \emptyset$. Then 
no vertex of~$W$ is $k$-connected to any vertex of~$T$.
Additionally, the set $Z \cap G(r)$ is a $k$-dominator in~$G(r)$ for 
every~$r \in W$ for which $Z \cap G(r)$ is not empty.
\end{lemma}
\begin{proof}
	By the definition of a \tscc, the vertices in~$T$ are strongly connected in
	$G \setminus Z$ but have no incoming edges from vertices of~$W$ in 
	$G \setminus Z$. Hence in~$G$ every path from a vertex of~$W$ to a vertex 
	of~$T$ contains an element of $Z$. 
	This implies that $Z \cap G(r)$ is a $k$-dominator in~$G(r)$ for every~$r \in W$. 
	For vertex connectivity we 
	have that the vertices in $W$ have no edges to the vertices in~$T$ in~$G$; 
	hence no vertex of~$W$ is $k$-connected to a vertex of~$T$.
\end{proof}

Let $r \in V$ be a vertex such that there does \emph{not} exist a set of 
vertices $S$ with 
$r \notin S$ that induces a \tscc in $G$, i.e., all vertices in $V$ can be 
reached from $r$. We show below that this is a sufficient condition such that 
whenever $Z$ is a $k$-dominator in the flow graph 
$G(r)$, then there exists a set $T$ that induces a \tscc in $G \setminus Z$ and 
\nt{T}{Z}{G} indeed exists. Thus if we only want to detect \tscc{s} and $k$-almost \tscc{s}
\nt{T}{Z}{G} for which both $T$ and $Z$ do not contain $r$, we can use the 
following approach to find one of them whenever one exists: 
\begin{compactenum}
	\item Search for a 
set of vertices $T$ with $r \notin T$ that induces a \tscc in $G$. 
\item If none is found, search for a $k$-dominator in $G(r)$. 
\item If a $k$-dominator $Z$ is found, find a set of vertices $T$ with 
$r \notin T$ that induces a \tscc in $G \setminus Z$.
\end{compactenum}
We formalize the correctness of this approach with the following two lemmata.
\begin{lemma}\label{lem:ktsccindom}
	Let $r$ be a vertex that can reach all vertices in $G$ and let $G(r)$ be 
	the flow graph rooted at $r$. Let $Z$ be a $k$-dominator in $G(r)$. Then
	there exists a set of vertices $T$ with $r \notin T$ that is a \tscc in
	$G \setminus Z$.
\end{lemma}
\begin{proof}
Let $D$ be the vertices dominated by $Z$ in $G(r)$. Since $D$ is dominated 
by~$Z$, there are no edges from vertices in $G(r) \setminus (D \cup Z)$ to 
vertices in~$D$ in~$G \setminus Z$. Thus either $D$ contains a \tscc that does not 
contain~$r$ in $G \setminus Z$ or there are vertices in $G \setminus G(r)$ that 
have edges to vertices in~$D$. By assumption $G \setminus G(r)$ is empty.
\end{proof}
\begin{lemma}\label{lem:kdomnewtscc}
Let $G(r)$ be a flow graph for some graph $G = (V,E)$ and some vertex $r \in V$. 
Let $Z$ be a $k$-dominator in $G(r)$ and let the set of vertices $T$ with 
$r \notin T$ induce a \tscc in $G \setminus Z$. Then \nt{T}{Z}{G} exists.
\end{lemma}
\begin{proof}
	Since $T$ induces a \tscc in $G \setminus Z$, the incoming edges of $T$ are
	clearly a subset of $Z$ for edge connectivity and the edges from $Z$ to $T$ 
	for vertex connectivity. It remains to show that $T$ has an incoming edge
	for each element of $Z$, i.e., that $Z$ satisfies the minimality condition
	for $k$-almost \tscc{s}. Assume by contradiction that \nt{T}{Z'}{H} exists for 
	a proper subset $Z'$ of $Z$. Then $r$ can reach $T$ only through the elements
	of $Z'$, i.e., $Z'$ is a $k$-dominator in $G(r)$, a contradiction to the 
	minimality of $k$-dominators.
\end{proof}

Now we want to generalize this approach to subgraphs $H$ of $G$, that is,
for some appropriately chosen root~$r$,
we first search for a \tscc in $H$ not containing $r$; if no \tscc is found, 
we search for a $k$-dominator in the flow graph $H(r)$; and 
for a $k$-dominator $Z$ in $H(r)$ we find a \tscc in $H \setminus Z$
not containing~$r$.
If no \tscc is found in $H$ but a $k$-dominator $Z$ is found, this yields a 
\nt{T}{Z}{H} by Lemma~\ref{lem:kdomnewtscc}. 
But when does a \nt{T}{Z}{H} imply a \nt{T}{Z}{G}?
First think of a \tscc in $H$ induced by a set of vertices $T$.
We are only allowed to look at the edges in $H$ but assume for now we know for
each vertex whether all its incoming edges in $G$ are also contained in $H$.
If some vertex of $T$ has more incoming edges in $G$ than in $H$, we cannot
decide whether $T$ also induces a \tscc in $G$. However, if for each vertex of~$T$
all its incoming edges in $G$ are also in $H$, then $T$ induces
a \tscc in $G$ as well.
Now think of a \nt{T}{Z}{H}. If the incoming edges of each vertex in $T$ are 
present in $H$, then also the elements of $Z$ have to be contained in $H$ because 
otherwise at least one vertex of $T$ would miss an incoming edge in $H$.
Thus by the observation for \tscc{s}
above we have that \nt{T}{Z}{G} exists. We state this formally
for slightly more general graphs $H$ in the following lemma and its corollary
such that we can apply these results to all graphs we define in this work.
Consider a graph $H = (V_H, E_H)$ constructed from a 
graph $G = (V, E)$ that is not necessarily a subgraph of $G$ but has the following 
guarantee for all vertices in some set $A \subseteq V_H \cap V$: 
For each vertex~$u$ of~$A$ the incoming edges of $u$ in $H$ are the same as in~$G$.
\begin{lemma}[Extension of Lemma~\ref{lem:tscc}]\label{lem:ktscc}
Let $G = (V, E)$ and $H = (V_H, E_H)$ be two graphs with the following guarantee 
for all vertices in $A \subseteq V_H \cap V$: For each 
vertex $u \in A$ we have $\In_H(u) = \In_G(u)$. Then a set of vertices 
$T \subseteq A$ induces a \tscc in $H$ if and only if it induces a \tscc in $G$.
\end{lemma}
\begin{proof}
By the guarantee on $A$, the incoming edges of the vertices in $T$ are the same 
in $H$ and $G$, i.e., we have $H[T] = G[T]$.
Furthermore, the set of vertices $T$ has no incoming edges in~$H$ if 
and only if it has no incoming edges in~$G$.
\end{proof}
\begin{corollary}[Extension of Corollary~\ref{cor:corr}]\label{cor:kcorr}
	Let $G$, $H$, and $A$ be as in Lemma~\ref{lem:ktscc}. For a set of vertices
	$T \subseteq A$ and a set of elements $Z$ we have that \nt{T}{Z}{H} exists
	if and only if \nt{T}{Z}{G} exists.
\end{corollary}
Our algorithms identify \tscc{s} and $k$-almost \tscc{s} that are subsets of $A$
in a graph $H$ with the guarantee described above. Together with 
Lemma~\ref{lem:knewtscc} this is crucial for the correctness of our approach.

To obtain algorithms with a good running time, we want to identify a \tscc or a
$k$-almost \tscc that is a subset of $A$ in $H$ whenever a \tscc or $k$-almost \tscc that
is a subset of $A$ (and is a $k$-isolated set) exists in $G$. Note that we cannot 
say anything about \tscc{s} in $H$ that contain vertices that are not in $A$. 
Recall that we find $k$-almost \tscc{s} by searching for $k$-dominators~$Z$ in flow 
graphs of $H$ and then searching for \tscc{s} in $H \setminus Z$.
Further, recall that since we want to achieve a running time of $o(mn)$ (at least 
for edge connectivity or for $k = 2$), we cannot afford to recompute all strongly
connected components and thus cannot assume that the graph the algorithm currently
operates on is strongly connected. Hence to argue about which $k$-almost \tscc{s} in $G$ 
we can identify by the search in a flow graph of $H$, we have to consider 
the reachability of 
a set $Z$ for which a \nt{T}{Z}{G} exists from the root of the flow graph in $H$.
The flow graphs used by our algorithms contain all vertices in $V_H \setminus A$.
Thus whenever some elements of a set $Z$ are 
not reachable from the root of such a flow graph, i.e., are not in $H(r)$,
then there exists a non-empty subset of $A$ that is not reachable from 
vertices in $V_H \setminus A$. Thus
in this case there exists a \tscc in $H$ that contains only vertices in $A$
and is also a \tscc in $G$ by Lemma~\ref{lem:ktscc}.
Hence in this case our algorithms can make progress by identifying in $H$ a \tscc that 
only contains vertices in $A$ instead of searching for the $k$-almost \tscc with 
respect to $Z$.
\begin{observation}\label{obs:notreach}
Let $G = (V, E)$, $H = (V_H, E_H)$, and $A$ be as in Lemma~\ref{lem:ktscc}. 
Let $H(r)$ be a flow graph in $H$ for some root~$r \in V_H$ 
that contains all vertices in $V_H \setminus A$.
If some vertices of $A$ are not 
contained in $G(r)$, then there exists a set of vertices $S \subseteq A$
that induces a \tscc $H[S]$ in $H$.
\end{observation}

We now come back to proper subgraphs of graphs.
Let $\subG{G}{h} = (\subV{h},\subE{h})$ be a subgraph of a graph 
$G = (V,E)$, i.e., $\subV{h} \subseteq V$ and $\subE{h}
\subseteq G[\subV{h}]$. We use the index $h$ to identify specific subgraphs
and the corresponding sets of edges and vertices.
Let $\good{G}{h}$ be the set of vertices in $\subG{G}{h}$
for which we can guarantee for each $u \in \good{G}{h}$ that
$\In_{\subG{G}{h}}(u) = \In_{G}(u)$ and let $\bad{G}{h}$ denote $\subV{h} 
\setminus \good{G}{h}$ (which is a subset of $V \setminus \good{G}{h}$).
If $\bad{G}{h} = \emptyset$, then $\subG{G}{h} = G$ and we can simply search for 
\scc{s} and $k$-separators in $G$ to make progress in our algorithms. Thus we
assume $\bad{G}{h} \ne \emptyset$ in the remainder of this subsection.
We next define graphs and flow graphs derived from $\subG{G}{h}$ in which we 
can identify a set $Z$ as a $k$-dominator in the flow graph whenever a 
\nt{T}{Z}{G} with $T \subseteq \good{G}{h}$ 
exists and each element of $Z$ is reachable from the root of the flow 
graph---except for one special case for vertex connectivity, which we consider
separately.
Intuitively, in the graphs derived from $\subG{G}{h}$ we view the vertices 
in $\bad{G}{h}$ as a (directed) clique. Since we do not know anything about the 
structure of the subgraph induced by $\bad{G}{h}$, this is a ``worst-case''
assumption on the connectivity of these vertices with respect to detecting
$k$-isolated sets.

We first define the graphs and flow graphs derived from $\subG{G}{h}$
for edge connectivity, then for vertex connectivity.
For edge connectivity we are able to identify a \nt{T}{Z}{G} with 
$T \subseteq \good{G}{h}$ for which the edges in $Z$ are reachable from 
the vertices in $\bad{G}{h}$ by using the following way of contracting 
the vertices in $\bad{G}{h}$ and use the new contracted vertex as the root in 
the flow graph.
\begin{definition}[Flow graph for edge connectivity]\label{def:kesubgraph}
Let $\subG{G}{h} = (\subV{h},\subE{h})$ be a subgraph of a graph 
$G = (V,E)$. Let $\good{G}{h}$ be the set of vertices in $\subG{G}{h}$
for which we can guarantee for each $u \in \good{G}{h}$ that 
$\In_{\subG{G}{h}}(u) = \In_{G}(u)$ and let $\bad{G}{h}$ denote $\subV{h} 
\setminus \good{G}{h}$. 

For edge connectivity and $\lvert \bad{G}{h} \rvert \ge 1$, we define the 
flow graph $\subGf{G}{h}(\rootf{G}{h})$ as follows.
Let the graph $\subGf{G}{h}$ be the multi-graph~$\subGf{G}{h} = 
(\subVf{h},\subEf{h})$ where all vertices in $\bad{G}{h}$ are contracted to a 
single vertex $\rootf{G}{h}$ in the following way.
The vertices $\subVf{h}$ are equal to $\good{G}{h} \cup \set{\rootf{G}{h}}$ and
the edges $\subEf{h}$ consists of all edges in $G[\good{G}{h}]$ and 
one edge $(u,\rootf{G}{h})$ for each edge in $(u,v) \in E \cap (\good{G}{h} 
\times \bad{G}{h})$ and, symmetrically, one edge $(\rootf{G}{h},v)$ for 
each $(u,v) \in E \cap (\bad{G}{h} \times \good{G}{h})$.
\end{definition}
With the flow graph $\subGf{G}{h}(\rootf{G}{h})$ we have the following
strategy to identify a \tscc or a $k$-almost \tscc in 
$\subG{G}{h}$ and $\subGf{G}{h}$ whenever a \tscc or a $k$-almost \tscc that
is a subset of $\good{G}{h}$ exists in $G$:
\begin{compactenum}
\item Search for a \tscc induced by vertices in $\good{G}{h}$ in $\subG{G}{h}$.
\item If none found, search for a $k$-dominator in $\subGf{G}{h}(\rootf{G}{h})$.
\item If a $k$-dominator~$Z$ is found, find a \tscc induced by 
vertices in $\good{G}{h}$ in $\subG{G}{h}\setminus Z$.
\end{compactenum} 
We detect a \tscc in the first step whenever some vertex of $\good{G}{h}$ is not 
reachable from any vertex of $\bad{G}{h}$. Thus if no such \tscc exists, all 
vertices of $\good{G}{h}$ are contained in $\subGf{G}{h}(\rootf{G}{h})$.
This has two consequences. First, if a $k$-dominator $Z$ is found in the 
second step, by Lemma~\ref{lem:kdomnewtscc} a \tscc induced by 
vertices $T \subseteq \good{G}{h}$ exists in $\subG{G}{h}\setminus Z$.
In this case \nt{T}{Z}{H} exists by Lemma~\ref{lem:kdomnewtscc}
and \nt{T}{Z}{G} exists by Corollary~\ref{cor:kcorr}; this is crucial for 
the correctness of the approach. Second, if for some set of vertices $T \subseteq 
\good{G}{h}$ and some set of edges $Z$ a \nt{T}{Z}{G} exists, then $Z$ 
is contained in $\subGf{G}{h}(\rootf{G}{h})$; thus in this case a
$k$-dominator is identified in $\subGf{G}{h}(\rootf{G}{h})$.

For vertex connectivity we have to use different flow graphs because we also want
to identify a \nt{T}{Z}{G} with $T \subseteq \good{G}{h}$ for which some of the 
vertices in $Z$ are contained
in $\bad{G}{h}$. We first consider the simpler case when $\lvert \bad{G}{h} 
\rvert \ge k$. In this case we connect an artificial root vertex to 
all vertices in $\bad{G}{h}$, which allows us to detect a set $Z$ for which 
a \nt{T}{Z}{G} exists when each vertex
in $Z$ is either contained in $\bad{G}{h}$ or reachable from a vertex in $\bad{G}{h}$.
When $0 < \lvert \bad{G}{h} \rvert < k$ we cannot use an additional vertex 
that we connect to the vertices of $\bad{G}{h}$ as root in the flow graph 
because the vertices of $\bad{G}{h}$ would be a $k$-dominator in this flow 
graph independent of the underlying graph $G$. We still want to detect a 
\nt{T}{Z}{G} with $T \subseteq \good{G}{h}$ for which $Z$ includes vertices 
in $\bad{G}{h}$ 
in this case. We further distinguish two cases: Either there
exists a vertex in $\bad{G}{h}$ that is not in $Z$ or $\bad{G}{h}$ is 
contained in $Z$. For the first case we use $\lvert \bad{G}{h} \rvert < k$
different flow graphs and we can identify the set of elements in $Z$ a flow 
graph for a vertex in $\bad{G}{h} \setminus Z$
when each vertex in $Z$ is either contained in $\bad{G}{h}$ or reachable from 
a vertex in $\bad{G}{h}$. In the second case, i.e., $B \subseteq Z$,
we cannot identify $Z$ in any of the flow graphs; we consider
this case explicitly by testing whether $\subG{G}{h}\setminus \bad{G}{h}$ is 
strongly connected and searching for a $(k-\lvert \bad{G}{h} \rvert)$-separator 
in $\subG{G}{h} \setminus \bad{G}{h}$.
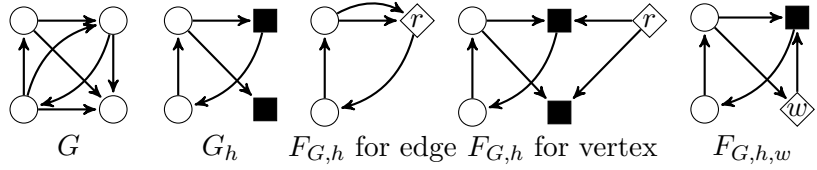
\begin{figure}
\centering
\begin{tikzpicture}
\tikzstyle{vertex}=[circle,draw,minimum size=3.5mm, solid]
\tikzstyle{arrow}=[->,line width=0.5pt,>=stealth',shorten >=1pt,thick]
\tikzstyle{bad}=[rectangle,draw,fill,minimum size=3mm, solid]
\tikzstyle{root}=[diamond,draw,inner sep=0pt, minimum size=4.5mm, solid]
\tikzstyle{ann} = [draw=none,fill=none]

\begin{scope}
\matrix[column sep=8mm, row sep=8mm]{
	\node[vertex] (1) {};
	& \node[vertex] (2) {};\\
	\node[vertex] (3) {};
	& \node[vertex] (4) {};\\
};
\path (1) edge[arrow] (2) edge[arrow] (4)
(2) edge[arrow, bend left, bend angle=20] (3) edge[arrow] (4)
(3) edge[arrow] (1) edge[arrow] (4) edge[arrow, bend left, bend angle=20] (2);
\node[ann,below=8mm]{$G$};
\end{scope}

\begin{scope}[xshift=2cm]
\matrix[column sep=8mm, row sep=8mm]{
	\node[vertex] (1) {};
	& \node[bad] (2) {};\\
	\node[vertex] (3) {};
	& \node[bad] (4) {};\\
};
\path (1) edge[arrow] (2) edge[arrow] (4)
(2) edge[arrow, bend left, bend angle=20] (3)
(3) edge[arrow] (1);
\node[ann,below=8mm]{$\subG{G}{h}$};
\end{scope}

\begin{scope}[xshift=4cm]
\matrix[column sep=8mm, row sep=8mm]{
	\node[vertex] (1) {};
	& \node[root] (2) {$r$};\\
	\node[vertex] (3) {};
	& \\
};
\path (1) edge[arrow] (2) edge[arrow, bend left, bend angle=50] (2)
(2) edge[arrow, bend left, bend angle=20] (3)
(3) edge[arrow] (1);
\node[ann,below=8mm]{$\subGf{G}{h}$ for edge};
\end{scope}

\begin{scope}[xshift=6.5cm]
\matrix[column sep=8mm, row sep=8mm]{
	\node[vertex] (1) {};
	& \node[bad] (2) {};
	& \node[root] (r) {$r$};\\
	\node[vertex] (3) {};
	& \node[bad] (4) {};
	& \\
};
\path (1) edge[arrow] (2) edge[arrow] (4)
(2) edge[arrow, bend left, bend angle=20] (3)
(3) edge[arrow] (1)
(r) edge[arrow] (2) edge[arrow] (4);
\node[ann,below=8mm]{$\subGf{G}{h}$ for vertex};
\end{scope}

\begin{scope}[xshift=9cm]
\matrix[column sep=8mm, row sep=8mm]{
	\node[vertex] (1) {};
	& \node[bad] (2) {};\\
	\node[vertex] (3) {};
	& \node[root] (4) {$w$};\\ 
};
\path (1) edge[arrow] (2) edge[arrow] (4)
(2) edge[arrow, bend left, bend angle=20] (3)
(3) edge[arrow] (1)
(4) edge[arrow] (2);
\node[ann,below=8mm]{$\subGf{G}{h,w}$};
\end{scope}
\end{tikzpicture}
\caption{Simple example graphs for graphs and vertices as in Definitions~\ref{def:kesubgraph} and~\ref{def:kvsubgraph}. 
The black square vertices denote vertices of $\bad{G}{h}$, the white 
diamond vertices denote the vertices that are used as roots in the 
corresponding flow graphs.}
\end{figure}
\begin{definition}[Flow graphs for vertex connectivity]\label{def:kvsubgraph}
Let $\subG{G}{h} = (\subV{h},\subE{h})$ be a subgraph of a graph 
$G = (V,E)$. Let $\good{G}{h}$ be the set of vertices in $\subG{G}{h}$
for which we can guarantee for each $u \in \good{G}{h}$ that 
$\In_{\subG{G}{h}}(u) = \In_{G}(u)$ and let $\bad{G}{h}$ denote $\subV{h} 
\setminus \good{G}{h}$. 

\begin{itemize}
\item For vertex connectivity and $\lvert \bad{G}{h} \rvert \ge k$, we define 
the flow graph $\subGf{G}{h}(\rootf{G}{h})$ as follows. Let $\subGf{G}{h}$ be 
the graph $\subG{G}{h}$ with an additional vertex~$\rootf{G}{h}$ and an 
additional edge from $\rootf{G}{h}$ to each vertex in~$\bad{G}{h}$.
\item  For vertex connectivity and $0 < \lvert \bad{G}{h} \rvert < k$, we 
define $\lvert 
\bad{G}{h} \rvert$ different flow graphs, one for each $w \in \bad{G}{h}$. Let 
$\subGf{G}{h,w}(\rootf{G}{h,w})$ denote the flow graph for $w \in \bad{G}{h}$. 
The root $\rootf{G}{h,w}$ is equal to $w$ and the graph $\subGf{G}{h,w}$
is the graph $\subG{G}{h}$ with an additional edge from $w$ to each vertex in
$\bad{G}{h} \setminus \set{w}$.
\end{itemize}
\end{definition}
The approach for vertex connectivity differs from the approach for 
edge connectivity mainly by (1) a different definition of the flow graph 
$\subGf{G}{h}(\rootf{G}{h})$ (2) in the case 
$0 < \lvert \bad{G}{h} \rvert < k$ by (2a) searching in $\lvert 
\bad{G}{h} \rvert$ different flow graphs and (2b) considering the special 
case to detect when a \nt{T}{Z}{G} with $Z \supseteq \bad{G}{h}$ exists.
We use the following lemma to show the correctness for the special case.
\begin{lemma}\label{lem:specialcase}
Consider vertex connectivity and assume $0 < \lvert \bad{G}{h} \rvert < k$.
Assume that \emph{no} \tscc induced by vertices in $\good{G}{h}$ exists in 
$\subG{G}{h}$ and that \emph{no} $k$-dominator exists in any of the flow graphs 
$\subGf{G}{h,w}(\rootf{G}{h,w})$ for $w \in \bad{G}{h}$. 

\begin{compactenum}[\lu a\ru]
\item If $\subG{G}{h} 
\setminus \bad{G}{h}$ is not strongly connected, let $T$ be a set of vertices
that induces a \tscc in $\subG{G}{h} \setminus \bad{G}{h}$ and let $Z' = \emptyset$.
\item If $\subG{G}{h} \setminus \bad{G}{h}$ is strongly connected, $\lvert 
\bad{G}{h} \rvert < k - 1$, and there exists a $(k - \lvert \bad{G}{h}
\rvert)$-separator~$Z'$
in $\subG{G}{h} \setminus \bad{G}{h}$, let $T$ be a set of vertices that induces a 
\tscc in $\subG{G}{h} \setminus (Z' \cup \bad{G}{h})$.
\end{compactenum}
Let $Z = Z' \cup \bad{G}{h}$.
If either the conditions for \enum{a} or for \enum{b} hold, then \nt{T}{Z}{G} exists.
\end{lemma}
\begin{proof}
	First note that $T \subseteq \good{G}{h}$, that there are edges from 
	$V \setminus T$ to $T$ in $\subG{G}{h}$, and that these edges are all from vertices 
	of $Z$. Thus it remains to show that 
	\emph{each} vertex of $Z$ has an edge to a vertex of $T$. 
	In Case~(b) the existence of a vertex in $Z'$ that has no edge to $T$ would 
	contradict the minimality of a $(k - \lvert \bad{G}{h} \rvert)$-separator (thus
	\nt{T}{Z'}{\subG{G}{h} \setminus \bad{G}{h}} exists). To see that also each 
	vertex of $\bad{G}{h}$
	has an edge to a vertex of $T$, assume by contradiction that there exists a 
	set of vertices $U \subseteq \bad{G}{h}$ such that there is no edge 
	from $U$ to $T$.
	By the assumption that no \tscc induced by vertices in $\good{G}{h}$ exists in 
$\subG{G}{h}$, we have that all vertices in $\good{G}{h}$ are reachable from some 
vertex of $\bad{G}{h}$ and thus contained in 
$\subGf{G}{h,w}(\rootf{G}{h,w})$ for all $w \in \bad{G}{h}$. Hence $Z \setminus 
U$ would dominate the vertices in $T$ in $\subGf{G}{h,u}(\rootf{G}{h,u})$
for $u \in U$, a contradiction to our assumptions.
\end{proof}

The following two corollaries together with Lemmata~\ref{lem:knewtscc}, 
\ref{lem:ktscc}, and~\ref{lem:specialcase} 
provide a summary of the results in this subsection.
\begin{corollary}\label{cor:kcorrdom}
Assume $\bad{G}{h}$ is such that the flow graph $\subGf{G}{h}(\rootf{G}{h})$
is defined.
If \emph{no} \tscc induced by vertices in $\good{G}{h}$ exists in $\subG{G}{h}$
and a set of elements~$Z$ with $\lvert Z \rvert < k$ is a $k$-dominator 
in $\subGf{G}{h}(\rootf{G}{h})$, 
then for some set of vertices $T \subseteq \good{G}{h}$ \nt{T}{Z}{G} exists.
The same holds for the flow graphs $\subGf{G}{h,w}(\rootf{G}{h,w})$ 
for $w \in \bad{G}{h}$ whenever they are defined.
\end{corollary}
\begin{proof}
Since no \tscc induced by vertices in $\good{G}{h}$ exists in $\subG{G}{h}$,
all vertices in $\good{G}{h}$ are reachable from some vertex of $\bad{G}{h}$
in $\subG{G}{h}$ and thus contained in $\subGf{G}{h}(\rootf{G}{h})$
(or for vertex connectivity and $\lvert \bad{G}{h} 
\rvert < k$ in $\subGf{G}{h,w}(\rootf{G}{h,w})$ for all $w \in \bad{G}{h}$).
Thus there exists a \tscc induced by a set of vertices $T \subseteq \good{G}{h}$
in $\subGf{G}{h} \setminus Z$ (Lemma~\ref{lem:ktsccindom}) and thus also in 
$\subG{G}{h} \setminus Z$ (Lemma~\ref{lem:ktscc}).
By Lemma~\ref{lem:kdomnewtscc} this implies that \nt{T}{Z}{\subG{G}{h}} exists.
Thus by Corollary~\ref{cor:kcorr} \nt{T}{Z}{G} exists.
\end{proof}

\begin{corollary}[Extension of Lemma~\ref{lem:finddom}]\label{cor:kfind}
Assume $\bad{G}{h}$ is such that the flow graph $\subGf{G}{h}(\rootf{G}{h})$
is defined. Let $T \subseteq \good{G}{h}$.
If \emph{no} \tscc induced by vertices in $\good{G}{h}$ exists in $G$,
for some set of elements~$Z$ with $\lvert Z \rvert < k$ \nt{T}{Z}{G} exists,
and $\bad{G}{h} \setminus Z$ is not empty, then $Z$ is a $k$-dominator 
in $\subGf{G}{h}(\rootf{G}{h})$.
The same holds for at least one of the flow graphs $\subGf{G}{h,w}(\rootf{G}{h,w})$ 
for $w \in \bad{G}{h}$ whenever they are defined.
\end{corollary}
\begin{proof} 
By Corollary~\ref{cor:kcorr} \nt{T}{Z}{\subG{G}{h}} exists.
Since no \tscc induced by vertices in $\good{G}{h}$ exists in $G$, no 
such \tscc exists in $\subG{G}{h}$ by Lemma~\ref{lem:ktscc} and thus all vertices in 
$\good{G}{h}$ and all elements in $Z$ are reachable from some vertex of $\bad{G}{h}$
in $\subG{G}{h}$ and hence contained in $\subGf{G}{h}(\rootf{G}{h})$
(or for vertex connectivity and $\lvert \bad{G}{h} 
\rvert < k$ in $\subGf{G}{h,w}(\rootf{G}{h,w})$ for all $w \in \bad{G}{h}$).
As we exclude the case that $\bad{G}{h} \setminus Z$ is empty 
(for vertex connectivity), we have that in the case $\lvert \bad{G}{h} 
\rvert < k$ there exists some  $w \in \bad{G}{h}$ such that $\rootf{G}{h,w} \notin Z$. 
Thus by Lemma~\ref{lem:knewtscc} $Z$ is a $k$-dominator in 
$\subGf{G}{h}(\rootf{G}{h})$, or for vertex connectivity and $\lvert \bad{G}{h} 
\rvert < k$, in $\subGf{G}{h,w}(\rootf{G}{h,w})$ for $w \in \bad{G}{h} \setminus Z$.
\end{proof}

\subsection{The algorithms for \kescc{s} and \kvscc{s}}\label{sec:kalgo}
In this subsection we present our algorithms based on specific subgraphs 
$\subG{G}{h}$ that allow us to identify \tscc{s} and $k$-almost \tscc{s} in $G$ that 
have at most a certain size. This is crucial for 
the runtime analysis of the hierarchical graph decomposition technique.
We first provide intuition for the runtime analysis and describe the algorithms
and then formally prove their correctness and running times.

Let $G = (V,E)$ be a simple directed graph. We consider the following 
hierarchical graph decomposition: For level $i \in \mathbb{N}$ let 
the subgraph $\subG{G}{i} = (V, E_i)$ of $G$ contain all vertices in $V$ and for 
each vertex of $V$ its first $2^i$ incoming edges in $E$ (for some arbitrary but
fixed ordering of the incoming edges of each vertex). Note 
that for $i \ge \log (\max_{v \in V}{\InDeg_G(v)})$ we have $\subG{G}{i} = G$.
Following the definitions in the previous subsection, 
let $\good{G}{i}$ be the set of vertices with in-degree at most $2^i$ in~$G$ and 
let $\bad{G}{i} = V \setminus \good{G}{i}$ be the set of vertices with in-degree 
more than $2^i$ in $G$.

Recall that we make progress in our algorithms by separating specific $k$-isolated sets,
namely \tscc{s} and $k$-almost \tscc{s}, from the remaining graph. The main idea
of the hierarchical graph decomposition is to detect ``vertices to separate''
that contain $O(2^i)$ vertices in $\subG{G}{i}$ in time proportional to
the number of edges in $\subG{G}{i}$, i.e., in time $O(n \cdot 2^i)$.
The search for ``vertices to separate'' is started at level $i = 1$. When the 
search is not successful, the level $i$ is increased by one. Thus if the 
level has to be increased up to $i^*$ to identify a set of vertices, then 
this set contains $\Omega(2^{i^*})$ vertices because otherwise it would have 
been detected already at level $i^* - 1$. The time spent in the levels 1 up to
$i^*$ forms a geometric series and thus can be bounded by $O(n \cdot 2^{i^*})$
if the work per level is $O(n \cdot 2^i)$.
In our algorithm we recurse on the identified $k$-isolated sets.
To account for the recursion, we search ``in parallel'' on $G$ and its reverse
graph $\revG{G}$, that is, we also search for \bscc{s} and $k$-almost \bscc{s} by searching 
for \tscc{s} and $k$-almost \tscc{s} in $\revG{G}$. 
The search for both ($k$-almost) top and bottom \scc{s} 
ensures that whenever a ($k$-almost) \tscc and a disjoint ($k$-almost) \bscc exist in~$G$, 
we only spend time proportional to the smaller one; that is, to
identify a set $S$ when the per-level runtime is $O(n \cdot 2^i)$, we spend
time $O(n \cdot \min\{\lvert S \rvert, \lvert V \setminus S \rvert\})$.

Let $T$ be a set of at most $2^i - k$ vertices such that there exists a 
\nt{T}{Z}{G} for some set of elements $Z$. Since the only
edges from vertices of $V \setminus S$ to $S$ are, for edge connectivity,
the edges in $Z$, or, for vertex connectivity, edges from $Z$, the in-degree
of each vertex in $T$ can be at most $2^i$. Thus the vertices in $T$ are 
contained in $\good{G}{i}$. This allows us to apply the results from the 
previous subsection to show that we can identify a \tscc or $k$-almost \tscc of $G$
with at most $2^i - k$ vertices
by searching for \tscc{s} in $\subG{G}{i}$ and for $k$-dominators in derived
flow graphs~$\subGf{G}{i}(\rootf{G}{i})$ as in 
Definitions~\ref{def:kesubgraph} and~\ref{def:kvsubgraph}.
\begin{lemma}[Extension of Lemma~\ref{lem:level}]\label{lem:klevel}
Let $G$ be a simple directed graph.
	\begin{compactenum}[\lu 1\ru]
		\item If a set of vertices $S$ with $\lvert S \rvert \le 2^i + 1$ 
		induces a \tscc $G[S]$ in $G$, 
		\label{sublem:klevelC}
		\item or if there is a set of elements $Z$ with $\lvert Z \rvert < k$ such that
		for some set of vertices~$S$ with $\lvert S \rvert 
		\le 2^i - k + 2$ there exists a $k$-almost \tscc $G[S]$ with 
		respect to~$Z$ in $G$, \label{sublem:klevelEV}
		\label{sublem:klevelb1}
	\end{compactenum}
	then $S \subseteq \good{G}{i}$.
\end{lemma}
\begin{proof}
	\begin{compactenum}[\lu 1\ru]
		\item Consider any graph $\tilde{G}$ in which $S$ induces a \tscc.
	Since a \tscc has no incoming edges, all incoming edges of vertices
	in $S$ have to come from other vertices of $S$. Thus each vertex in $S$ can 
	have an in-degree of at most $\lvert S \rvert - 1$ in~$\tilde{G}$.
	The claim follows for $\tilde{G} = G$.
		\item Consider Case~\enum{1} for $\tilde{G} = G \setminus Z$.
		In $G$ each vertex in $S$ can have 
		at most $k-1$ additional incoming edges compared to $G \setminus Z$, namely 
		an edge from each vertex of $Z$ for vertex connectivity and the 
		edges in~$Z$ for edge connectivity.
		Thus we can bound the in-degree in~$G$ of each vertex in~$S$ by~$2^i$. We 
		have $S \subseteq \good{G}{i}$.
	\end{compactenum}
\end{proof}

Let $\mmdeg$ be the minimum of $\max_{v \in V}{\InDeg_G(v)}$ and 
$\max_{v \in V}{\OutDeg_G(v)}$. In Algorithm~\ref{algk} we start the search 
for ($k$-almost) top \scc{s} at $i = 1$. Whenever the search is not successful, we
increase~$i$ by one, until we have $2^i \ge \mmdeg$, that is, 
$\subG{G}{i} = G$ or $\subG{\revG{G}}{i} = \revG{G}$. For the search the 
Procedure~\ref{algklevel} is used as long as $2^i < \mmdeg$,
i.e., both $\bad{G}{i}$ and $\bad{\revG{G}}{i}$ are not empty, and the 
Procedure~\ref{algksearch} is used afterwards. If a \tscc of $G$ or $\revG{G}$ is
identified, the procedures return the set of vertices in the \tscc (and an 
empty set $Z$). If a $k$-almost \tscc \nt{S}{Z}{G} or \nt{S}{Z}{\revG{G}} is identified, 
the procedures return the sets $S$ and $Z$. When one of the procedures
returns a non-empty set $S$ and a (potentially empty) set $Z$, the algorithm 
recurses on each of $G[S]$ and $G[V \setminus S]$ for edge connectivity and 
on each of $G[S \cup Z]$ and $G[V \setminus S]$ for vertex connectivity.
As the set $Z$ is not used further for edge connectivity, we do not have to 
worry about the direction of the edges in $Z$. We define the following shortcuts.
\begin{definition}
For a set of vertices $S$ and a set of elements $Z$, let $\GS$ for 
vertex connectivity be equal to~$G[S \cup Z]$ and for edge connectivity equal 
to $G[S]$. Let $\GV$ be equal to $G[V \setminus S]$.
\end{definition}
If all the calls to the procedures return only empty sets, then the considered 
graph $G$ is a \kscc and is returned.
\begin{algorithm}
\SetAlgoRefName{\kscc}
\caption{k-edge or k-vertex strongly connected components}
\label{algk}
\SetKwProg{myproc}{}{}{}
\myproc{$\reck(G):$}{
	$\mmdeg \leftarrow \min{(\max_{v \in V}{\InDeg_G(v)}, \max_{v \in V}{\OutDeg_G(v)})}$\;
	\For{$i \leftarrow 1$ \KwTo $\lceil \log \mmdeg \rceil - 1$}{
			$(S, Z) \leftarrow \klevelsearch(G, i)$\;
			\If{$S \ne \emptyset$}{
				\Return{$\reck(\GS) \cup \reck(\GV)$}\;
			}
	}
	$(S, Z) \leftarrow \kallsearch(G)$\;
	\eIf{$S \ne \emptyset$}{
		\Return{$\reck(\GS) \cup \reck(\GV)$}\;
	}{
	\Return{$\set{G}$}\;
	}
}
\end{algorithm}

We use Procedure~$\topsccex(H, B)$ to denote the search for a \tscc induced 
by vertices $S$ in a graph~$H$ such that $S$ does not contain a vertex of~$B$. 
Such a \tscc can simply be found by marking \tscc{s} in a standard \scc algorithm.
We let all procedures that search for an \scc return the set of vertices in the 
\scc.

The Procedure~\ref{algklevel} first constructs the graph $\subG{\G}{i} \in 
\set{\subG{G}{i}, \subG{\revG{G}}{i}}$ and searches 
for a \tscc in $\subG{\G}{i}$ that does not contain a vertex of $\bad{\G}{i}$, i.e.,
only contains vertices of $\good{G}{i}$ for which we know that all their 
incoming edges in $\G \in \set{G, \revG{G}}$ are present in $\subG{\G}{i}$.
If such a \tscc is found, its set of vertices is returned (and an empty 
set is returned as the second return value).
Otherwise the procedure continues in one of two ways.
First consider edge connectivity and the case
$\lvert \bad{\G}{i} \rvert \ge k$ for vertex connectivity.
Then the flow graph $\subGf{\G}{i}(\rootf{\G}{i})$ is constructed and searched 
for $k$-dominators. To find $k$-dominators the known algorithms listed in 
Section~\ref{sec:kscc_short} are used. 
If a $k$-dominator~$Z$ is found, a \tscc in~$\subG{\G}{i} 
\setminus Z$ that does not contain a vertex of 
$\bad{\G}{i}$ is found and the
vertices in the \tscc and the set $Z$ are returned.
Otherwise the procedure returns two empty sets.
In the case $\lvert \bad{\G}{i} \rvert < k$ for vertex connectivity
for each $w \in \bad{\G}{i}$ the flow graph $\subGf{\G}{i,w}(\rootf{\G}{i,w})$
is constructed and searched for a $k$-dominator. If at least one of the 
flow graphs contains a $k$-dominator $Z$, then a \tscc in~$\subG{\G}{i} 
\setminus Z$ that does not contain a vertex of 
$\bad{\G}{i}$ is found and the
vertices in the \tscc and the set $Z$ are returned.
If none of the flow graphs contains a $k$-dominator, there could still exist a 
$k$-separator that contains all vertices of $\bad{\G}{i}$, which we also want 
to detect at this level. To this end, we first test whether $\subG{\G}{i} 
\setminus \bad{\G}{i}$ is strongly connected and return the vertices in a \tscc 
of $\subG{\G}{i} \setminus \bad{\G}{i}$ if not. 
If $\subG{\G}{i} \setminus \bad{\G}{i}$ is strongly connected and 
$\lvert \bad{\G}{i} \rvert < k-1$, we search for 
a $(k-\lvert \bad{\G}{i} \rvert)$-separator $Z'$ in $\subG{\G}{i} \setminus 
\bad{\G}{i}$, using the known algorithm to find $k$-separators
as described in Section~\ref{sec:kscc_short}.
If such a $Z'$ exists, then $Z = Z' \cup \bad{\G}{i}$ is a $k$-separator and 
we return the vertices in a \tscc of $\subG{\G}{i} \setminus Z$ and the set $Z$.
\begin{procedure}
\caption{kIsolatedSetLevel($G$, $i$)} 
\label{algklevel}
\ForEach{$\G \in \set{G, \revG{G}}$}{
	\tcc{$2^i < \max_{v \in V}{\InDeg_\G(v)} \Longrightarrow \bad{\G}{i} \ne \emptyset$}
	construct $\subG{\G}{i} = (V, E_i)$ with 
	$E_i = \cup_{v\in V}\{\text{first } 2^i \text{ edges in } \In_\G(v)\}$\;
	$\bad{\G}{i} = \{v \mid \InDeg_\G(v) > 2^i\}$\;
	$S \leftarrow \topsccex(\subG{\G}{i}, \bad{\G}{i})$\label{l:tscc}\;
	\If{$S \ne \emptyset$}{
		\Return{$(S, \emptyset)$}\;
	}
	\If{$\lvert \bad{\G}{i} \rvert \ge k$ or \kescc{s} searched for}{
		construct flow graph $\subGf{\G}{i}(\rootf{\G}{i})$\;
		\If{exists $k$-dominator $Z$ in $\subGf{\G}{i}(\rootf{\G}{i})$}{
			$S \leftarrow \topsccex(\subG{\G}{i}\setminus Z, \bad{\G}{i})$\label{l:kdom1}\;
			\Return{$(S, Z)$}
		}
	}
	\Else(\tcc*[h]{only when \kvscc{s} searched for and $\lvert \bad{\G}{i} \rvert < k$}){
		construct flow graphs $\subGf{\G}{i,w}(\rootf{\G}{i,w})$ 
		for all $w \in \bad{\G}{i}$\;
		\If{exists $k$-dominator $Z$ in some $\subGf{\G}{i,w}(\rootf{\G}{i,w})$}{
			$S \leftarrow \topsccex(\subG{\G}{i}\setminus Z, \bad{\G}{i})$\label{l:kdom2}\;
			\Return{$(S, Z)$}
		}
		$S \leftarrow \topscc(\subG{\G}{i} \setminus \bad{\G}{i})$\label{l:tsccB}\;
		\If{$S \subsetneq V \setminus \bad{\G}{i}$}{
			\Return{$(S, \bad{\G}{i})$}\;
		}
		\If{$\lvert \bad{\G}{i} \rvert < k-1$ and 
		exists $(k-\lvert \bad{\G}{i} \rvert)$-separator $Z'$ in 
		$\subG{\G}{i} \setminus \bad{\G}{i}$}{
			$Z \leftarrow Z' \cup \bad{\G}{i}$\;
			$S \leftarrow \topscc(\subG{\G}{i}\setminus Z)$\label{l:kBsep}\;
			\Return{$(S, Z)$}
		}
	}
}
\Return{$(\emptyset, \emptyset)$}\;
\end{procedure}

If Procedure~\ref{algklevel} returns only empty sets for all levels~$i 
< \log \mmdeg$, Procedure~\ref{algksearch} is called. It first tests
whether the graph~$G$ is strongly connected and returns a \tscc of $G$ if not.
If $G$ is strongly connected, it searches for a $k$-separator in $G$.
If a $k$-separator~$Z$ is found, disjoint top and 
bottom \scc{s} exist in $G \setminus Z$. The procedure returns a \tscc 
of $G \setminus Z$ and the set $Z$ in this case.
If the graph~$G$ is strongly connected and does not contain a
$k$-separator, then $G$ is a \kscc. In this case the 
procedure returns two empty sets, this branch of recursion stops in 
Algorithm~\ref{algk}, and $\reck(G)$ returns $G$.
\begin{procedure}
\caption{kIsolatedSet($G$)} 
\label{algksearch}
$S \leftarrow \topscc(G)$\;
\If{$S \subsetneq V$}{
	\Return{$(S, \emptyset)$}\;
}
\If{exists $k$-separator $Z$ in $G$}{
	$S \leftarrow \topscc(G \setminus Z)$\;
	\Return{$(S, Z)$}\;
}
\Return{$(\emptyset, \emptyset)$}\;
\end{procedure}

\subsubsection{Correctness}
Note that, for each of $\G \in \set{G, \revG{G}}$, $\subG{\G}{i}$ is a 
subgraph of~$\G$ such that all vertices of $\subG{\G}{i}$ that are missing 
incoming edges in~$\subG{\G}{i}$ are contained in $\bad{\G}{i}$.
Recall that the flow graph $\subGf{\G}{i}(\rootf{\G}{i})$ is specified in 
Definitions~\ref{def:kesubgraph} and~\ref{def:kvsubgraph} and further that
for a non-empty set $S$ and a (potentially empty) set $Z$ returned by 
Procedure~\ref{algklevel} or~\ref{algksearch} $\GV$ is equal to 
$G[V \setminus S]$ and $\GS$ is equal to $G[S]$ for edge connectivity and equal 
to $G[S \cup Z]$ for vertex connectivity.

To show the correctness of the algorithms for \kscc{s}, the following three parts
are needed:
\enum{1}~Every step in the algorithm can be executed as described.
\enum{2}~Whenever a non-empty set~$S$ and a set~$Z$ are identified, 
every \kscc of~$G$ is completely contained in either $\GS$ or $\GV$ and both
$\GS$ and $\GV$ are not empty and thus proper subgraphs of~$G$.
\enum{3}~Whenever no non-empty set~$S$ is identified in~$G$, then $G$ is a 
\kscc.

We introduce Lemma~\ref{lem:ksearch} and Theorem~\ref{th:kcorr} to show 
Parts~(2) and~(3), respectively.
For Part~\enum{1} we need in particular that \enum{1a} $\bad{G}{i}$ and $\bad{\revG{G}}{i}$
are not empty whenever Procedure~$\klevelsearch(G, i)$ is called and
that \enum{1b} whenever a $k$-dominator in a flow graph derived from $\subG{\G}{i}$
is identified, then there exists a \tscc in $\subG{\G}{i}\setminus Z$ that does 
not contain a vertex of $\bad{\G}{i}$. For \enum{1a} the algorithm ensures 
$i < \log \mmdeg$, i.e., $2^i < \max_{v \in V}{\InDeg_G(v)}$ and
$2^i  < \max_{v \in V}{\OutDeg_G(v)}$. Lemma~\ref{lem:ktsccindom} implies \enum{1b}.

The following lemma shows, based on the results of Subsection~\ref{sec:ksubgraph},
that the Procedures~\ref{algklevel} and~\ref{algksearch} 
indeed find \tscc{s} or $k$-almost \tscc{s} in $G$ or $\revG{G}$. This will imply
by Lemma~\ref{lem:knewtscc}
that every \kscc of~$G$ is completely contained in either $\GS$ or $\GV$
whenever $\reck$ is called recursively on $\GS$ and $\GV$.
\begin{lemma}\label{lem:ksearch}
	Let $0 < i < \log \mmdeg$. If \enum{1} Procedure~\ref{algklevel} or
	\enum{2} Procedure~\ref{algksearch}
	return a non-empty set $S$ and a set $Z$ \lu with $\lvert Z \rvert < k$\ru
	for graph $G$,
	then $S$ induces either a top or a bottom \scc or a $k$-almost top or bottom
	\scc with respect to $Z$ in $G$.
	Additionally, the set $V \setminus (S \cup Z)$ is not empty.
\end{lemma}
\begin{proof}
	\begin{compactenum}[\lu 1\ru]
		\item Let $\G \in \set{G, \revG{G}}$ be the graph in which $S$ and $Z$ are 
		identified. In Procedure \ref{algklevel} a set of vertices $S$ and a 
		set of elements $Z$ can be 
		identified in three ways: \enum{a} $\G[S]$ is a \tscc in $\subG{\G}{i}$ 
		that does not include vertices of $\bad{\G}{i}$ and $Z = \emptyset$; 
	\enum{b} the algorithm finds a $k$-dominator $Z$ in 
	$\subGf{\G}{i}(\rootf{\G}{i})$ or $\subGf{\G}{i,w}(\rootf{\G}{i,w})$ for some
	$w \in \bad{\G}{i}$, and $\G[S]$ is a \tscc in~$\subG{\G}{i} \setminus Z$ that 
	does not include a vertex of~$\bad{\G}{i}$; 
	or \enum{c} vertex connectivity is
	considered, we have $\lvert \bad{\G}{i} \rvert < k$, and either \enum{c1} 
	$\subG{\G}{i} \setminus \bad{\G}{i}$ is not strongly connected or \enum{c2} there 
	exists a $(k - \lvert \bad{\G}{i} \rvert)$-separator $Z'$ in 
	$\subG{\G}{i} \setminus \bad{\G}{i}$ for $k - \lvert \bad{\G}{i} \rvert > 1$. Let 
	$Z' = \emptyset$ in Case~\enum{c1} and let $Z = Z' \cup \bad{\G}{i}$.
	Then in Case~\enum{c} $S$ induces a \tscc 
	in $\subG{\G}{i} \setminus Z$.
	We use that Case~\enum{b} can only occur if Case~\enum{a} did not occur
	and that Case~\enum{c} can only occur if the Cases~\enum{a} 
	and~\enum{b} did not occur in this call to Procedure \ref{algklevel}.
	With this we have:
	In Case~\enum{a} $\G[S]$ is a \tscc in $\G$ by Lemma~\ref{lem:ktscc}.
	In Case~\enum{b} \nt{S}{Z}{\G} exists by Corollary~\ref{cor:kcorrdom}.
	In Case~\enum{c} \nt{S}{Z}{\G} exists by Lemma~\ref{lem:specialcase}.
	
	In the Cases~\enum{a} and \enum{b} we have 
	$\bad{\G}{i} \cap V \setminus (S \cup Z) \ne \emptyset$ and thus 
	both $S$ and $V \setminus (S \cup Z)$ are not empty. In Case~\enum{c1} 
	$V \setminus (S \cup Z) = \emptyset$ is explicitly avoided by testing whether
	$S$ is a proper subset of $V \setminus Z$. In Case~\enum{c2}
	$V \setminus (S \cup Z) \ne \emptyset$ follows from the definition of 
	$k$-separators.
	
	\item
	In Procedure~\ref{algksearch} a non-empty set of vertices~$S$ and a set of elements 
	$Z$ can be identified in two ways:
	\enum{a} $G[S]$ is a top \scc in~$G$ with~$S \ne V$ and~$Z = \emptyset$; 
	or \enum{b} $G$ is strongly connected, the algorithm finds a $k$-separator~$Z$
	in~$G$ and~$G[S]$ is a $k$-almost top \scc w.r.t.\ $Z$ in~$G$. 
	In both cases it clearly holds that $V \setminus (S \cup Z) \ne \emptyset$
	because $G\setminus Z$ contains a top and a bottom \scc that are disjoint and
	only the \tscc is contained in $G[S]$.
	\end{compactenum}
\end{proof}

\begin{theorem}[Correctness]\label{th:kcorr}
Let $G$ be a simple directed graph. $\reck(G)$ computes the \kscc{s} of $G$.
\end{theorem}

\begin{proof}
We denote by $G$ the graph for which $\reck(G)$ is called, which is equal
to the input graph for the initial call to $\reck$ and equal to a subgraph
of the input graph for all recursive calls to $\reck$.
Whenever $\reck(G)$ identifies a set $S$ and recursively calls itself 
on $\GS$ and $\GV$, then by Lemma~\ref{lem:ksearch} the set $S$ either \enum{a} induces 
a top or bottom \scc in $G$ or \enum{b} 
a $k$-almost top or bottom \scc with respect to the set $Z$ in $G$. 
	In Case~\enum{a} each \kscc is completely contained in either~$\GS$ or~$\GV$
	by the fact that every \kscc is strongly connected.
	In Case~\enum{b} by Lemma~\ref{lem:knewtscc} no vertex 
	of~$V \setminus (S \cup Z)$ is $k$-connected to any vertex of~$S$ in~$\G$;
	thus each \kscc is completely contained in either~$\GS$ or~$\GV$
	by the definition of a \kscc.
This implies that each \kscc of $G$ is a \kscc in $\GS$ or $\GV$.
Note that each \kscc of a subgraph of $G$ is by 
definition a subset of a \kscc of $G$. Hence the union of the \kscc{s} of $\GS$ 
and $\GV$ yields the \kscc{s} of $G$.
Furthermore, by Lemma~\ref{lem:ksearch} both $S$ and $V \setminus (S \cup Z)$
are not empty, i.e., the recursion is on proper subgraphs of $G$ and thus the
algorithm terminates.
	
To show that the algorithm does not terminate until it has correctly identified
all \kscc{s} of the input graph, it remains to show that whenever $\reck(G)$ 
is called, it 
either identifies a set $S$ and recursively calls itself on $\GS$ and $\GV$
or the graph $G$ is strongly connected and does not contain a $k$-separator, i.e.,
the graph $G$ is a \kscc. For this it is sufficient to show that if the 
algorithm has not identified a set~$S$ in the for-loop, then it either 
identifies a set by the call to Procedure~\ref{algksearch} or $G$ is
strongly connected and does not contain a $k$-separator. 
If $G$ is not strongly connected, then Procedure~\ref{algksearch} 
returns a \tscc in $G$. If $G$ is strongly connected, then
either $G$ does not contain $k$-separators or Procedure~\ref{algksearch} finds a 
$k$-separator~$Z$ and identifies a \tscc in $G \setminus Z$, 
which exists by the definition of a $k$-separator.
\end{proof}

\subsubsection{Runtime}
We already showed by Lemma~\ref{lem:klevel} that sets of at most $2^i - k + 2$
vertices that induce a \tscc or a $k$-almost \tscc in $G$ are contained in 
$\good{G}{i} = V \setminus \bad{G}{i}$, 
i.e., the incoming edges of each vertex in such a set are present in $\subG{G}{i}$.
This allows us to apply the results from Subsection~\ref{sec:ksubgraph} to show 
that our algorithms find such a set whenever one exists.
Let $\mmdeg$ be the minimum of $\max_{v \in V}{\InDeg_G(v)}$ and 
$\max_{v \in V}{\OutDeg_G(v)}$ as in Algorithm~\ref{algk}.
\begin{lemma}[Extension of Lemma~\ref{lem:find}]\label{lem:kfind}
	If for some integer $1 \le i < \log \mmdeg$ and 
	$\G \in \set{G, \revG{G}}$ there exists a set of vertices $T \subseteq 
	\good{\G}{i}$ that either induces a \tscc
	or a $k$-almost \tscc with respect to some set of elements $Z$ with 
	$\lvert Z \rvert < k$ and $T \subsetneq V \setminus Z$ in $\G$, 
	then $\klevelsearch(G, i)$ returns a non-empty set $S$.
\end{lemma}
\begin{proof}
	If $\G[T]$ is a \tscc in $\G$, then $\subG{\G}{i}[T]$ is a \tscc without vertices of
	$\bad{\G}{i}$ in $\subG{\G}{i}$ by Lemma~\ref{lem:ktscc}. Thus in this case
	in Line~\ref{l:tscc} of Procedure~\ref{algklevel} a non-empty set $S$ is returned.
	In the following assume that no non-empty set $S'$ is identified in 
	Line~\ref{l:tscc}, i.e.,
	that no \tscc without vertices of $\bad{\G}{i}$ exists in $\subG{\G}{i}$.
	
	If $\G[T]$ is a $k$-almost \tscc with respect to a set of elements~$Z$ in 
	$\G$, then either $\bad{\G}{i} \setminus Z$ is empty or a non-empty
	set $S$ is by Corollary~\ref{cor:kfind} returned 
	in Line~\ref{l:kdom1} or, for vertex connectivity
	and $\lvert \bad{\G}{i} \rvert < k$, in Line~\ref{l:kdom2}.
	
	If $\bad{\G}{i} \setminus Z$ is empty, i.e., $Z \supseteq \bad{\G}{i}$, assume
	that no $k$-dominator is found in any flow graph, i.e., no non-empty set $S'$
	is returned before Line~\ref{l:tsccB}.
	If $Z = \bad{\G}{i}$, then $\subG{\G}{i}[T]$ is a \tscc in $\subG{\G}{i} 
	\setminus \bad{\G}{i}$ and a non-empty set $S$ 
	is identified in Line~\ref{l:tsccB}. Otherwise $Z' = Z \setminus \bad{\G}{i} 
	\ne \emptyset$ contains less than $k - \lvert \bad{\G}{i} \rvert$ vertices
	and \nt{T}{Z'}{\subG{\G}{i}\setminus \bad{\G}{i}} exists. If no non-empty 
	set was identified in 
	Line~\ref{l:tsccB}, then $\subG{\G}{i} \setminus \bad{\G}{i}$ is strongly 
	connected. Thus, since $T \subsetneq V \setminus Z$, the set $Z'$ is a 
	$(k - \lvert \bad{\G}{i} \rvert)$-separator in $\subG{\G}{i} \setminus \bad{\G}{i}$.
	Hence a non-empty set $S$ is returned in Line~\ref{l:kBsep}.
\end{proof}
\begin{corollary}\label{cor:ksize}
Let \enum{a} $i^* < \log \mmdeg$ be the smallest integer such that for all
$1 \le i < i^*$ $\klevelsearch(G, i)$ returned empty sets $S'$ and 
a non-empty set $S$ is returned by $\klevelsearch(G, i^*)$ 
or, (2) if all calls to $\klevelsearch(G, i)$ returned empty sets $S'$ for $1 \le 
i < \log \mmdeg$, but $\kallsearch(G)$ returned a non-empty set $S$, let $i^* = 
\lceil \log \mmdeg \rceil$. Then $\lvert S \rvert > 2^{i^*-1} - k + 2$, respectively.
\end{corollary}
\begin{proof}
	By Lemma~\ref{lem:ksearch} the set $S$ either induces a \tscc or a $k$-almost
	\tscc with $S \subsetneq V \setminus Z$ in $\G$, where $\G$ is 
	either $G$ or $\revG{G}$ and $Z$ is the second argument returned by 
	$\klevelsearch(G, i^*)$ or $\kallsearch(G)$. By Lemma~\ref{lem:kfind} we have 
	$S \cap \bad{\G}{i^*-1} \ne \emptyset$. The claim then follows from
	Lemma~\ref{lem:klevel}.
\end{proof}

In the runtime analysis we argue that whenever the search for a 
set~$S$ stops at a certain level~$i^*$, the algorithm has spent time at most
proportional to the number of vertices in the smaller set of~$S$ and~$V \setminus S$.
Additional to Corollary~\ref{cor:ksize}, we use
the fact that if a graph is not strongly connected, it contains a top and a 
bottom \scc that are disjoint, which implies that one of them contains at most
half of the vertices of the graph.

\begin{theorem}[Runtime \kescc{s}]\label{th:ktimeE1}
	Algorithm~\ref{algk} for \kescc{s} can be implemented in time $O(n^2 \log n)$
	for any integral constant $k > 2$ and in time $O(n^2)$ for $k = 2$.
\end{theorem}
\begin{proof}
We denote by $n$ and $m$ the number of vertices and edges in the input graph
and by $n'$ the number of vertices in the graph $G$
of the current level of the recursion in Algorithm~\ref{algk}.
Let $\mmdeg$ denote the minimum of $\max_{v \in V}{\InDeg_G(v)}$ and 
$\max_{v \in V}{\OutDeg_G(v)}$. Let $S$ be a non-empty set of vertices returned by
$\klevelsearch$ or $\kallsearch$.

To efficiently construct the graphs $\subG{\G}{i}$ for 
$1 \le i < \lceil \log \mmdeg \rceil$ and $\G \in \set{G, \revG{G}}$
we maintain for all vertices $w$ a list of $\In(w)$ and a list of $\Out(w)$. 
We do not update this data structure immediately when $G$ is split into $G[S]$ 
and $G[V \setminus S]$ but remove obsolete entries in $\In(w)$ and $\Out(w)$
whenever we encounter them while constructing $\subG{\G}{i}$. This can happen
at most once for each entry and thus takes total time $O(m) \in O(n^2)$.

We first analyze the time per iteration of the for-loop in $\reck(G)$.
In iteration~$i$ the algorithm calls $\klevelsearch(G, i)$. Searching on both $G$ and 
$\revG{G}$ only increases the running time by a factor of two; we
analyze the search on $G$, the search on $\revG{G}$ is analogous. 
Given the above data structure, constructing $\subG{G}{i}$, $\subGf{G}{i}$, and 
determining $\bad{G}{i}$ takes time $O(2^i \cdot n')$. Finding \scc{s} in 
$\subG{G}{i}$ takes time linear in the number of edges in $\subG{G}{i}$, i.e., 
time $O(2^i \cdot n')$.
Finding $k$-dominators in $\subGf{G}{i}(\rootf{G}{i})$ takes time 
$O(2^i \cdot n' \cdot \log n')$ for $k > 2$ and time $O(2^i \cdot n')$ for $k=2$.
We give below the analysis for $k > 2$, the analysis for $k=2$ is identical but
without the log-factor. We have that $\klevelsearch(G, i)$ and hence 
iteration~$i$ of the outer for-loop take time $O(2^i \cdot  n' \cdot \log n')$.

The search for \scc{s} and $k$-separators in $\kallsearch$ can also be done in 
time proportional to the number of edges times $\log n$, which can be bounded 
by $O(\mmdeg \cdot n' \cdot \log n')$.

Let $i^* < \lceil \log \mmdeg \rceil$ be the last iteration of $\reck(G)$, i.e., 
the iteration 
before $\reck(G)$ returns $G$ (Case~\enum{1}) or recursively calls itself 
on the subgraphs $G[S]$ and $G[V \setminus S]$ (Case~\enum{2}). 
The time spent in the iterations $i = 1$ up to $i^*$
forms a geometric series that can be bounded by $O(2^{i^*} \cdot n' \cdot \log n')$. 
If $i^* =  \lceil \log \mmdeg \rceil - 1$, we have $\mmdeg \le 2^{i^* + 1}$ and
thus also the time spent in $\kallsearch$ can be bounded with $O(2^{i^*} \cdot n' 
\cdot \log n')$.

Case~\enum{1}: 
If $\reck(G)$ returns~$G$, the recursion stops. In this case the running
time is proportional to $ \mmdeg \cdot n' \cdot \log n' \le (n')^2 \cdot
\log n'$.

Case~\enum{2}:
In iteration $i^* - 1$ no top or bottom \scc or $k$-almost top or bottom \scc was 
detected in $G_{i^* - 1}$. Thus we have by
Corollary~\ref{cor:ksize} that $\lvert S \rvert > 2^{i^*-1} - k + 2$.
Let $\G \in \set{G, \revG{G}}$ be the graph in which $S$ was detected
at level $i^*$. If \enum{a} $G$ is strongly connected, then $\G[S]$ is a $k$-almost 
\tscc in $\G$ with respect to some set of elements $Z$ and we have that 
there exists a $k$-almost \bscc in $\G$ with respect to $Z$. This
$k$-almost \bscc is contained in $V \setminus S$ and by 
Lemmata~\ref{lem:klevel} and~\ref{lem:kfind} contains more than $2^{i^*-1} - k + 
2$ vertices, i.e., we have $\lvert V \setminus S \rvert > 2^{i^*-1} - k + 2$.
If \enum{b} $G$ is not strongly connected, then there exist a top and a bottom
\scc in $G$ which are disjoint and both have more than $2^{i^*-1} - k + 2$ 
vertices by Lemmata~\ref{lem:klevel} and~\ref{lem:kfind}. Since $G[S]$ is a 
strongly connected subgraph
of $G$, it has to be completely contained in either the top or the bottom \scc.
Thus also in this case we have $\lvert V \setminus S \rvert > 2^{i^*-1} - k + 2$.
Let $\min(\lvert S \rvert, \lvert V \setminus S \rvert)$ be denoted by $n_e$.
By the definition of $n_e$ we have $n_e \le n'/2$. By the analysis above, 
without the recursive calls, $\reck(G)$ spends time proportional to $n_e \cdot n' 
\cdot \log n'$.

We show next that the total time spent in all recursive calls to $\reck(G)$,
and thus the total running time of Algorithm~\ref{algk} for \kescc{s}, is of
order $f(n) = 2 n^2 \log n$. In Case~\enum{1}, i.e., if the recursion stops,
we have $f(n') = (n')^2 \log n < 2 (n')^2 \log n$. In Case~\enum{2} we have by
induction, and in particular for $n' = n$,
\begin{align*}
	f(n') &\le f(n_e) + f(n' - n_e) + n_e n' \log{n'} \,,\\
	&= 2 n_e^2 \log{n'} + 2 (n'-n_e)^2 \log{n'} +  n_e n' \log{n'}\,,\\
	&= 2 n_e^2 \log{n'} + 2 (n')^2 \log{n'} - 4 n_e n' \log{n'} + 2 n_e^2 \log{n'}
	+ n_e n' \log{n'}\,,\\
	&= 2 (n')^2 \log{n'} + 4 n_e^2 \log{n'} - 3 n_e n' \log{n'} \,,\\
	&\le 2 (n')^2 \log{n'} \,,
\end{align*}
where the last inequality follows from $n_e \le n'/2$.
\end{proof}

\begin{theorem}[Runtime \kvscc{s}]\label{th:ktimeV1}
	Algorithm~\ref{algk} for \kvscc{s} can be implemented in time $O(n^3)$
	for any integral constant $k > 2$ and in time $O(n^2)$ for $k = 2$.
\end{theorem}
\begin{proof}
The proof is analogous to the proof of Theorem~\ref{th:ktimeE1} except
that \enum{1} the recursion is on $G[S \cup Z]$ and $G[V \setminus S]$ whenever 
a $k$-almost top or bottom \scc with respect to $Z$ was found and 
\enum{2} the time to find $k$-dominators and $k$-separators is proportional 
to the number of edges times the number of vertices in the graph for $k > 2$.

To efficiently update the adjacency lists, we update the lists of the vertices 
in~$Z$ for each of $G[S \cup Z]$ and $G[V \setminus S]$ immediately when 
the recursive call occurs, while for all other vertices the same argument as 
before applies.
%

As before, let $n'$ denote the number of vertices in the graph $G$
in the current level of recursion in Algorithm~\ref{algk}. 
Let $n_v = \min(\lvert S \rvert, \lvert V \setminus S \rvert - k + 1)$.
We first analyze the runtime for an arbitrary integral constant $k > 2$ and 
then for $k = 2$. Let 
$O(n_v \cdot (n')^2)$ be the runtime bound
of $\reck(G)$ without the recursive calls whenever the recursion does not stop.
Note that $n_v \le n' / 2$.
We stop the recursion when the number of vertices $n'$ is less than $14 k^3$,
i.e., a constant; this can happen at most $n$ times. In this case we can use the 
known $O(m n^2)$-time algorithm for \kvscc{s} to compute the \kvscc{s} of $G$ in 
constant time. We obtain an upper bound of order $f(n) = n^3$ 
on the total runtime of the algorithm for \kvscc{s} as follows:
\begin{align*}
	f(n') &\le f(n_v + k) + f(n' - n_v) + n_v (n')^2 \,,\\
	&= (n_v + k)^3 + (n' - n_v)^3 + n_v (n')^2 \,,\\
	&=  (n')^3 - 2 n_v (n')^2 + 3 n_v^2 n' + 3 n_v^2 k + 3  n_v k^2 
	+ k^3\,,\\
	&\le (n')^3 - 2 n_v (n')^2 + 3 n_v^2 n' + 7 k^3 n_v^2 \,,\\
	&\le (n')^3 \,.
\end{align*}

For $k = 2$ let $O(n_v \cdot n')$ be the runtime bound
of $\reck(G)$ without the recursive calls whenever the recursion does not stop.
We stop the recursion whenever the
number of remaining vertices is at most $8$, i.e., a constant. In this
case we use the $O(m n)$-algorithm to determine the \vscc{s} of $G$
in constant time.
With $3 n' / 4 \ge 6 n_v / 5$ and $n' / 4 \ge 2 > 9 / 5$ we have 
$5 n'  \ge 6 n_v + 9$, which we use in the following to obtain
an upper bound of order $f(n) = 3 n^2$ on the total running time of 
Algorithm~\ref{algk} for \kvscc{s}.
We have
\begin{align*}
	f(n') &\le f(n_v + 1) + f(n' - n_v) + n_v n' \,,\\
	&= 3(n_v + 1)^2 + 3(n' - n_v)^2 + n_v n' \,,\\
	&= 3n_v^2 + 6n_v + 3+ 3(n')^2 - 6n_v n' + 3n_v^2 
	+ n_v n' \,,\\
	&= 3(n')^2 + 6n_v^2 - 5n_v n' + 6n_v + 3\,,\\
	&\le 3(n')^2 \,.
 	\qedhere
\end{align*}
\end{proof}

\section{An $O(m^2 / \log n)$-time algorithm for \escc{s}}\label{sec:local}
In this section we combine our results for almost \tscc{s} and dominators in 
subgraphs in Appendix~\ref{sec:ksubgraph} with 
a local-search technique used for B\"uchi games by
Chatterjee~et~al.~\cite{ChatterjeeJH03}.
Throughout the section we only consider edge-connectivity and $k = 2$. We use 
the definitions in Sections~\ref{sec:prelim} and~\ref{sec:kscc_short} but use
the more common terms \emph{bridges} and \emph{edge-dominators} instead of 
2-separators and 2-dominators.
In the algorithm we assume that each vertex in the input graph has constant in- 
and out-degree. We show how every graph with $m$ edges can be transformed in 
$O(m)$ time into a graph with in- and out-degree at most three with equivalent 
\escc{s}, $\Theta(m)$ vertices, and $\Theta(m)$ edges. We are not aware of such a 
transformation for \vscc{s}. 

\begin{algorithm}
\SetAlgoRefName{2eSCC-sparse} 
\caption{2-edge strongly connected components in $O(m^2 / \log n)$ time}
\label{alge2}
\SetKw{Stop}{stop}
\SetKwInOut{Input}{Input}
\SetKwData{True}{true}
\SetKwData{False}{false}
\SetKwData{Success}{success}
\Input{graph $G$ with $\InDeg(v) \le 3$ and $\OutDeg(v) \le 3$ for all $v \in V$}
$\lost \leftarrow \lceil \log n \rceil$\;
\Repeat{$J = \emptyset$}{
	find \scc{s} $C_1, \ldots, C_c$ of $G$\;
	remove edges between \scc{s} of $G$\;
	$J \leftarrow \emptyset$ \tcc*[f]{set of vertices that lost edges in this iteration}\;
	\For{$i \leftarrow 1$ \KwTo $c$}{
		$X \leftarrow \allbridges(C_i)$\;
		$G \leftarrow G \setminus X$\;
		add all vertices adjacent to the edges in $X$ to $J$\;
	}
	\If{$0 < \vert J \rvert < \lost$}{
		\Repeat{$S = \emptyset$ or $\lvert J \rvert \ge \lost$}{
			$S \leftarrow \localalg(G, J)$\label{l:localcall}\;
			remove edges between $S$ and $V \setminus S$ from $G$\;
			add vertices that lost edges to $J$\;
		}
	}
}
\Return{$C_1, \ldots, C_c$}
\end{algorithm}
We first describe the algorithm and then prove its correctness and running time.
In Algorithm~\ref{alge2} the main repeat-until loop 
without the local searches is equivalent to the simple
$O(mn)$-algorithm for \escc{s}: First the strongly connected components are 
determined, then in each strongly connected component all bridges 
are identified and removed. Both operations take time $O(m)$. 
This is repeated until no more bridges are found.
Our algorithm removes the bridges and the edges between the strongly 
connected components from the graph and stores all vertices that were adjacent
to the removed edges in the set $J$. The set $J$ is empty
at the beginning of each iteration of the main repeat-until loop.
If no bridges are found in an iteration of the main repeat-until loop, the
algorithm terminates. In this case each strongly connected component of the 
maintained graph~$G$ is 2-edge strongly connected.
If $J$ contains less than $\lost$ vertices, where $\lost$ is of order $\log n$,
then Procedure~$\localalg$ is called. $\localalg(G, J)$ 
identifies a non-empty set of vertices that induces a top or bottom \scc 
or an almost top or bottom \scc in $G$ whenever there is one that contains a vertex in $J$ 
and has less than $\depth$ vertices, where $\depth$ is also of order $\log n$. 
If the procedure returns a non-empty set $S$,
then the edges between~$S$ and the remaining vertices $V \setminus S$ are 
removed from $G$ and the vertices that lost edges are added to $J$.
The call to $\localalg(G, J)$ is repeated until either $J$ contains more than 
$\lost$ vertices or $\localalg(G, J)$ returns an empty set. Then a new iteration
of the main repeat-until loop is started.
Since vertices are added to $J$ only when they lost edges and $J$ is reset in
each iteration of the main repeat-until loop, the set $J$ can contain $\Omega(\log n)$ 
vertices only $O(m / \log n)$ times. If $\localalg(G, J)$ returns an empty 
set, we can show that each top or bottom \scc or almost top or bottom \scc in 
the current graph $G$ contains $\Omega(\log n)$ vertices.
In this case in the following iteration of the main repeat-until loop either 
at least two sets of vertices with $\Omega(\log n)$ vertices each are separated
from each other (by deleting the edges between them) or the algorithm
terminates. Thus also this case can happen at most $O(m / \log n)$ times.
Hence the total time without the calls to Procedure~$\localalg$ can be bounded
with $O(m^2 / \log n)$. The constant-degree assumption allows us 
to bound the work done in Procedure~$\localalg$. 

\begin{procedure}
\caption{2IsolatedSetLocal($G$, $J$)} 
$\depth \leftarrow \lceil \varepsilon \log n \rceil$ for some $\varepsilon \in (0,1)$\;
\ForEach{$\jj \in J$}{
	\ForEach{$\G \in \set{G,\revG{G}}$}{
		find set of vertices $\subVG{\G}{\jj}$ with distance at most $\depth$ to $\jj$ 
		in~$\G$ with $\bfs$ from $\jj$ in $\revG{\G}$\;
		$\subG{\G}{\jj} \leftarrow \G[\subVG{\G}{\jj}]$\;
		$\bad{\G}{\jj} \leftarrow \set{ v \in \subVG{\G}{\jj} \mid \In_\G(v) \cap ((V 
		\setminus \subVG{\G}{\jj}) \times \set{v}) \ne \emptyset}$\tcc*[f]
		{set of vertices with incoming edges from $V \setminus \subVG{\G}{\jj}$}\;
		$T \leftarrow \topsccex(\subG{\G}{\jj}, \bad{\G}{\jj})$\;
		\tcc{$\bad{\G}{\jj} = \emptyset \Rightarrow$ no edge from $V 
		\setminus \subVG{\G}{\jj}$ to $\subVG{\G}{\jj} \Rightarrow T \ne \emptyset$}
		\If{$T \ne \emptyset$}{
			\If{exists edge from $T$ to $V \setminus T$ in $\G$}{
				\Return{$T$}\label{ll:tscc}\;
			}
			\Else(\tcc*[f]{$T = \subVG{\G}{\jj}$ and 
			$\subG{\G}{\jj}$ is \tscc and \bscc in $\G$}){
				\If{exists bridge $e$ in $\subG{\G}{\jj}$}{
					$U \leftarrow \topscc(\subG{\G}{\jj}\setminus \set{e})$\;
					\Return{$U$}\label{ll:bridge}\;
				}
			}
		}
		\Else{
			let $\subGe{\G}{\jj}$ be $\subG{\G}{\jj}$ with $\bad{\G}{\jj}$ contracted to $\roote{\G}{\jj}$\tcc*[f]{see Definition~\ref{def:kesubgraph}}\;
			\If{exists edge-dominator $e$ in $\subGe{\G}{\jj}(\roote{\G}{\jj})$}{
				$U \leftarrow \topsccex(\subG{\G}{\jj}\setminus \set{e}, \bad{\G}{\jj})$\;
				\Return{$U$}\label{ll:dom}\;
			}
		}
	}
}
\Return{$\emptyset$}
\end{procedure}
The local searches in Procedure~$\localalg$ use the results presented 
in Appendix~\ref{sec:ksubgraph} on subgraphs~$\subG{G}{\jj}$ defined as follows.
Let $\jj$ be a vertex of $J$ and let $\depth$ be a parameter set to $\lceil 
\varepsilon \log n \rceil$ for some $\varepsilon \in (0,1)$.
Let $\subVG{G}{\jj}$ be the set of vertices that have
a path of length at most $\depth$ to $\jj$ in $G$. The subgraph~$\subG{G}{\jj}$
is equal to $G[\subVG{G}{\jj}]$. Following the definitions in 
Appendix~\ref{sec:ksubgraph}, the set of vertices $\bad{G}{\jj}$ is defined as all 
vertices in $\subVG{G}{\jj}$ that have incoming edges from $V \setminus \subG{G}{\jj}$,
i.e., for each vertex in $\good{G}{\jj} = \subVG{G}{\jj} \setminus \bad{G}{\jj}$
we know that all its incoming edges in $G$ are contained in the subgraph 
$\subG{G}{\jj}$. As in Definition~\ref{def:kesubgraph}, let $\subGe{G}{\jj}$ be the 
graph $\subG{G}{\jj}$ with all vertices in $\bad{G}{\jj}$ contracted to a new vertex
$\roote{G}{\jj}$. We use the flow graph $\subGe{G}{\jj}(\roote{G}{\jj})$ to find 
edge-dominators and almost \tscc{s} with respect to them.

In Procedure~$\localalg(G, J)$ we search for a \tscc or an almost \tscc in
$\subG{\G}{\jj}$ for each $\jj \in J$ and each of $\G \in \set{G, \revG{G}}$.
As Algorithm~\ref{alge2} makes progress by removing edges from 
the maintained graph~$G$, we only want to identify (almost) \tscc{s} for which
$G$ contains edges between the (almost) \tscc and the remaining graph.
Note that it can easily happen that there exists a \tscc in $G$ without outgoing
edges (i.e. a \tscc that is also a \bscc) because no vertex is removed from $J$ when 
Procedure~$\localalg(G, J)$ is called repeatedly in the inner repeat-until loop
of Algorithm~\ref{alge2}.
When Procedure~$\localalg(G, J)$ returns a non-empty set of vertices~$S$ (that 
induces a (almost) top or bottom \scc in $G$), the edges between the vertices in $S$
and the vertices in $V \setminus S$ are removed from $G$.
We now describe the steps in Procedure~$\localalg(G, J)$.
First the vertices in $\subVG{\G}{\jj}$ are identified by running a breadth-first 
search of depth~$\depth$ from $\jj$ on $\revG{\G}$. Then a \tscc induced by a set 
of vertices $T$ contained in $\good{\G}{\jj} = \subVG{\G}{\jj} \setminus \bad{\G}{\jj}$ 
is searched in $\subG{\G}{\jj}$. 
If $T$ exists and $\G$ contains edges from $T$ to 
$V \setminus T$, then $T$ is returned. If $T$ exists but there are no edges
from $T$ to $V \setminus T$ in $\G$, then $T$ also induces a \bscc in $\G$.
Since all vertices in $\subVG{\G}{\jj}$ are reachable from $\jj$, we have $T = 
\subVG{\G}{\jj}$ and $\bad{\G}{\jj} = \emptyset$ in this case.
Thus we \enum{a} cannot make progress by separating $T$ from $V \setminus T$
and \enum{b} cannot use the flow graph $\subGe{\G}{\jj}(\roote{\G}{\jj})$ for identifying
almost \tscc{s}. However, since for each vertex in $\subVG{\G}{\jj}$ all its
incoming edges are contained in $\subG{\G}{\jj}$ and $\subG{\G}{\jj}$ is strongly
connected, an almost \tscc in $\subG{\G}{\jj}$ can be found whenever one exists
by searching for a bridge~$e$ in $\subG{\G}{\jj}$ and a \tscc in $\subG{\G}{\jj}
\setminus \set{e}$. If a bridge~$e$ exists, the procedure returns the vertices
in the \tscc. If no bridge exists, then $\subG{G}{\jj}$ is a \escc and we do not
have to consider the vertices in $\subVG{G}{\jj}$ further.
If no \tscc induced by vertices in $\good{\G}{\jj}$ exists in $\subG{\G}{\jj}$,
then each vertex in $\subG{\G}{\jj}$ can be reached from vertices of $\bad{\G}{\jj}$
and $\bad{\G}{\jj}$ is not empty. In this case we search for an edge-dominator
in $\subGe{\G}{\jj}(\roote{\G}{\jj})$. If an edge-dominator~$e$ exists, the procedure
returns the vertices in a \tscc in 
$\subG{\G}{\jj} \setminus \set{e}$ that is induced by vertices of $\good{\G}{\jj}$.

We show next how to transform a graph to a constant degree graph without 
changing the \escc{s}. The transformation basically replaces each vertex~$v$ 
such that $\maxdeg{v} = \max(\InDeg(v),\OutDeg(v)) > 3$ by two cycles of $\maxdeg{v}$
many vertices, one cycle in each direction.
\begin{lemma}\label{lem:constdeg}
	Let $G = (V, E)$ be a simple directed graph with an arbitrary fixed ordering in 
	$\In(v)$ and $\Out(v)$ for each $v \in V$. Let $\maxdeg{v} = 
	\max(\InDeg(v),\OutDeg(v))$. Let $\widetilde{G}$ be $G$ with every vertex~$v 
	\in V$ with $\maxdeg{v} > 3$ replaced by $\maxdeg{v}$ vertices $V_v$ as 
	follows. Let $v_0, \ldots, v_{\maxdeg{v}-1}$ be the vertices of $V_v$. We add
	for each $v_i$ an edge $(v_i,v_{(i + 1) \bmod \maxdeg{v}})$ and an edge $(v_i, 
	v_{(i - 1) \bmod \maxdeg{v}})$ to $\widetilde{G}$. Additionally we add for 
	each edge $(u, v) \in E$ that is the $i$-th edge in $\Out(u)$ and the $j$-th 
	edge in $\In(v)$ an edge $(u_{i-1}, v_{j-1})$ to~$\widetilde{G}$. 
	For vertices $v \in V$ with $\maxdeg{v} \le 3$, let $V_v = \set{v}$.
	Then $G[S]$ 
	is a \escc in $G$ if and only if $\widetilde{G}[\cup_{u \in S} V_u]$ is a \escc in $\widetilde{G}$.
\end{lemma}
\begin{proof}
	Since every subgraph $\widetilde{G}[V_v]$ is 2-edge strongly connected, it is
	sufficient to show that there are two edge-disjoint paths
	from a vertex $u$ to a vertex $v \ne u$ in $G$ if and only if there are two
	edge-disjoint paths from some vertex $\widetilde{u} \in V_u$ to some 
	vertex $\widetilde{v} \in V_v$ with $v \ne u$ in $\widetilde{G}$. We can assume
	w.l.o.g.\ that the edge-disjoint paths are simple.
	
	$\Leftarrow :$ Let $\widetilde{P}_1$ and $\widetilde{P}_2$ be two edge-disjoint 
	paths in $\widetilde{G}$ 
	from a vertex $\widetilde{u} \in V_u$ to a vertex $\widetilde{v} \in V_v$ with
	$v \ne u$. We can construct two edge-disjoint paths $P_1$ and $P_2$ from $u$ to $v$
	in $G$ by simply removing all edges that are contained in a subgraph 
	$\widetilde{G}[V_w]$ for some $w \in V$ and replacing all vertices 
	$\widetilde{w} \in V_w$ with $w$ in the remaining edges.
	As there is a one-to-one relation between edges between subgraphs 
	$\widetilde{G}[V_w]$ and edges in $E$, the paths $P_1$ and $P_2$ must be 
	edge-disjoint.
	
	$\Rightarrow :$ Let $P_1$ and $P_2$ be two edge-disjoint paths from a vertex $u$ 
	to a vertex $v \ne u$ in $G$. We construct two edge-disjoint paths 
	$\widetilde{P}_1$ and $\widetilde{P}_2$ from a vertex $\widetilde{u} \in V_u$ 
	to a vertex 
	$\widetilde{v} \in V_v$ in $\widetilde{G}$. First add for all edges in the path
	$P_\ell$, $\ell \in \set{1,2}$, the corresponding edges between two different 
	subgraphs $\widetilde{G}[V_w]$ in $\widetilde{G}$ to $\widetilde{P}_\ell$. It remains to connect 
	the edges in $\widetilde{P}_\ell$ within subgraphs $\widetilde{G}[V_w]$ with
	$\maxdeg{w} > 3$. For any~$w$ that is only contained in one of~$P_1$ and~$P_2$,
	we can select some arbitrary path within $\widetilde{G}[V_w]$ that connects
	the edges in $\widetilde{P}_\ell$. Let $w$ be a vertex in both~$P_1$ and~$P_2$.
	Let $w_0, \ldots, w_{\maxdeg{w}-1}$ be the vertices of $V_w$. We connect the path 
	$\widetilde{P}_1$ in the subgraph $\widetilde{G}[V_w]$ by using the edges in one
	direction, i.e., from $w_i$ to $w_{(i + 1) \bmod \maxdeg{w}}$ for 
	$0 \le i < \maxdeg{w}$, and we connect the path 
	$\widetilde{P}_2$ in the subgraph $\widetilde{G}[V_w]$ by using the edges 
	in the other direction, i.e., from 
	$w_i$ to $w_{(i - 1) \bmod \maxdeg{w}}$. In this way the paths $\widetilde{P}_1$ and 
	$\widetilde{P}_2$ are edge-disjoint.
\end{proof}

\subsection{Correctness}
The correctness of the main repeat-until loop in Algorithm~\ref{alge2} apart 
from Procedure~$\localalg$ follows from the correctness of the basic algorithm 
for computing \escc{s}. To show the correctness of Algorithm~\ref{alge2}
including Procedure~$\localalg$, the following two parts are needed:
\enum{1}~Every step in Procedure~$\localalg$ can be executed as described.
\enum{2}~Whenever a non-empty set~$S$ is identified by Procedure~$\localalg(G, J)$,
every \escc of~$G$ is completely contained in either $G[S]$ or $G[V \setminus S]$
and there are edges between $S$ and $V \setminus S$ in $G$. The latter property
ensures that the algorithm terminates.

Let $\G \in \set{G, \revG{G}}$.
For Part~\enum{1} we need in particular that \enum{1a} whenever we consider the 
flow graph $\subGe{\G}{\jj}(\roote{\G}{\jj})$, the set $\bad{\G}{\jj}$ contains 
at least one vertex; and that \enum{1b} whenever there is an 
edge-dominator~$e$ in $\subGe{\G}{\jj}(\roote{\G}{\jj})$, 
then there exists a \tscc without vertices of $\bad{\G}{\jj}$ in $\subG{\G}{\jj} 
\setminus \set{e}$. Lemma~\ref{lem:ktsccindom} implies \enum{1b}.
For \enum{1a} note that if $\bad{\G}{\jj} = \emptyset$, we have that 
there are no edges from vertices of $V \setminus \subVG{\G}{\jj}$ to vertices 
of $\subVG{\G}{\jj}$ in $\G$. Thus in this case the subgraph $\subG{\G}{\jj}$ has 
to contain a \tscc.

\begin{lemma}\label{lem:localsearch}
	If Procedure~$\localalg$ returns a non-empty set of vertices $S$, 
	then \enum{a}~there exists an edge between vertices of $S$ and vertices 
	of $V \setminus S$ in~$G$ and \enum{b}~each \escc of $G$ is completely 
	contained in either $G[S]$ or $G[V \setminus S]$.
\end{lemma}
\begin{proof}
Let $\G \in \set{G, \revG{G}}$ be the graph in which the returned set is 
identified. Procedure~$\localalg$ returns a non-empty set for some $\jj \in J$ in 
three cases. It first searches for a set of vertices $T$ with 
$T \subseteq \good{\G}{\jj}$ that induces a \tscc $\G[T]$ in $\subG{\G}{\jj}$. 
If the search is successful, 
it determines whether there exists an edge between vertices of $T$ and 
vertices of $V \setminus T$ in~$G$. If this is satisfied, it returns the 
set $T$ (Case~\enum{1}).
If this is not satisfied, i.e., $G[T]$ is a top and a bottom \scc in $G$,
but $G[T]$ contains a bridge~$e$,
then the procedure returns a set of vertices $U$ that induces a top or 
bottom \scc $G[U]$ in~$G[T]\setminus \set{e}$ (Case~\enum{2}). If 
$\subG{\G}{\jj}$ does not contain a \tscc without a vertex of $\bad{\G}{\jj}$, i.e., 
all vertices in $\good{\G}{\jj}$ are reachable from some vertex of $\bad{\G}{\jj}$,
then the procedure searches for an edge-dominator
in~$\subGe{\G}{\jj}(\roote{\G}{\jj})$; if an edge-dominator~$e$ exists,
the procedure returns a set of vertices~$U$ with $U \subseteq 
\good{\G}{\jj}$ that induces a \tscc $\G[U]$ in $\subG{\G}{\jj} \setminus 
\set{e}$ (Case~\enum{3}).

In Case~\enum{1} the existence of edges between vertices of $T$ and vertices of
$V \setminus T$ is explicitly checked. By Lemma~\ref{lem:ktscc} $\G[T]$ is a 
\tscc in $\G$. Since every \escc is strongly connected, each \escc 
of~$G$ is completely contained in either $G[T]$ or $G[V \setminus T]$.

In Case~\enum{2} $\G[U]$ is an almost \tscc in $\G$ with 
respect to~$e$ by the definition of a bridge. In Case~\enum{3} we have by Lemma~\ref{lem:kdomnewtscc} and 
Corollary~\ref{cor:kcorr} that $\G[U]$ is an almost \tscc in $\G$ 
with respect to~$e$. Thus in both cases 
by Lemma~\ref{lem:knewtscc} no vertex of $V \setminus U$ has two edge-disjoint 
paths to any vertex of $U$ in $\G$. Hence each \escc of~$G$ has to be 
completely contained in either $G[U]$ or $G[V \setminus U]$.
The identified edge~$e$ has one endpoint in $V \setminus U$ and one in $U$, 
i.e., there exists at least one edge between vertices of $V \setminus U$ and 
vertices of $U$ in $G$.
\end{proof}

\begin{theorem}[Correctness]\label{th:corrE2}
Algorithm~\ref{alge2} computes the \escc{s} of the input graph.
\end{theorem}
\begin{proof}
	Algorithm~\ref{alge2} repeatedly removes edges from the input graph until
	the \scc{s} in the remaining graph correspond to the \escc{s} of the input
	graph. We show the correctness of the algorithm by showing that 
	\enum{a} the removed edges cannot be in a \escc, \enum{b} when the
	algorithm terminates, each \scc is a \escc, and \enum{c} the algorithm 
	terminates.
	
	By definition, \escc{s} are strongly connected subgraphs that do not contain
	a bridge. Thus \enum{a} clearly holds when edges between \scc{s} or bridges
	are removed in the repeat-until loop without the calls to $\localalg$. 
	By Lemma~\ref{lem:localsearch}, \enum{a}~also holds 
	for the edges removed after $\localalg$ returns a non-empty set of vertices.
	
	To show~\enum{b}, first note that whenever $\localalg$ is called,
	there will be another iteration of the repeat-until loop. This is because
	the algorithm terminates only in the case that $J$ is empty but
	$\localalg$ is only called when there are vertices in $J$ and $\localalg$
	does not remove vertices from $J$. Consider the last iteration of the
	repeat-until loop. In this iteration no bridges were identified as otherwise 
	$J$ cannot be empty. Thus no \scc 
	in $G$ contains a bridge, i.e., each \scc in $G$ is a \escc.
	
	For~\enum{c} we show that in each iteration of the inner and the outer-repeat 
	until loop either edges are removed or the algorithm terminates. The inner 
	repeat-until loop terminates when Procedure~$\localalg$ returns an empty set.
	Whenever Procedure~$\localalg$ returns a non-empty set~$S$, by 
	Lemma~\ref{lem:localsearch} there exist edges
	between $S$ and $V \setminus S$ in~$G$, which are then removed from $G$.
	For the outer repeat-until loop we have that the algorithm terminates 
	if no bridges are identified in the for-loop; otherwise at least the 
	bridges are removed from~$G$.
\end{proof}

\subsection{Runtime}

The next lemma applies the results of Appendix~\ref{sec:ksubgraph} to the
subgraphs used in Algorithm~\ref{alge2} to show that we can indeed find all desired
subgraphs with at most $\depth$ vertices in $\localalg$.
\begin{lemma}\label{lem:newtsccball}
If Procedure~$\localalg(G, J)$ returns an empty set, 
then
\begin{compactenum}[\lu a\ru]
\item each top or bottom \scc that is not disconnected from the remaining graph
\item and each almost top or almost bottom \scc in $G$ 
\end{compactenum}
that contains a vertex in $J$, has at least $d$ vertices.
\end{lemma}
\begin{proof}
We show the lemma by showing that if by contradiction an (almost) top or bottom 
\scc as described would exist, then Procedure~$\localalg(G, J)$ would return
a non-empty set.

Recall that $\subVG{\G}{\jj}$ is the set of vertices that can reach a vertex~$\jj$ 
in~$\G \in \set{G, \revG{G}}$ using a path containing at most $\depth$ edges 
and let $\subG{\G}{\jj} = \G[\subVG{\G}{\jj}]$.
Further recall that $\bad{\G}{\jj}$ is the set of vertices of $\subVG{\G}{\jj}$ with 
incoming edges from vertices of $V \setminus \subVG{\G}{\jj}$ and that $\good{\G}{\jj}$
is equal to $\subVG{\G}{\jj} \setminus \bad{\G}{\jj}$. The flow 
graph~$\subGe{\G}{\jj}(\roote{\G}{\jj})$ is as in Definition~\ref{def:kesubgraph}.

Assume there exists a set of vertices~$S$ with $\jj \in S$ and $\lvert S 
\rvert \le \depth$ that induces \enum{a}~a~\tscc $\G[S]$ in~$\G$ that has 
edges to vertices in $V \setminus S$ in $\G$ or \enum{b} an 
almost \tscc~$\G[S]$ in $\G$ with respect to some edge~$e$.
Since $\lvert S \rvert \le \depth$ and $\jj \in S$, any simple path in $\G[S]$ 
from a vertex of $S \setminus \set{\jj}$ to $\jj$ can contain at most $\depth-1$ edges. 
Thus all vertices with incoming edges to vertices of $S$ in $\G$ can reach $\jj$
using a path with at most $\depth$ edges in $\G$. Hence $S$ is a subset of 
$\subVG{\G}{\jj}$ and all incoming edges of $S$ are contained in $\subG{\G}{\jj}$, i.e., 
$S \subseteq \good{\G}{\jj}$.

In Case~\enum{a} $\G[S]$ is a \tscc in 
$\subG{\G}{\jj}$ by Lemma~\ref{lem:ktscc} and the procedure returns 
a non-empty set in Line~\ref{ll:tscc}.

For Case~\enum{b} assume that no set is returned in Line~\ref{ll:tscc}, i.e.,
there does \emph{not} exist a \tscc induced by a set of vertices $T$ in $\good{\G}{\jj}$
in $\subG{\G}{\jj}$ such that there are edges from $T$ to $V \setminus T$ in $\G$.
First consider the case that there exists a \tscc induced by a set of vertices 
$T$ in $\good{\G}{\jj}$ in $\subG{\G}{\jj}$ such that there are \emph{no}
edges from $T$ to $V \setminus T$ in $\G$. In this case $T$ is a top and a bottom
\scc in $G$ by Lemma~\ref{lem:ktscc}. Since all vertices in $\subVG{\G}{\jj}$
are reachable from $\jj$, we have $T = \subVG{\G}{\jj}$ and $\bad{\G}{\jj} = \emptyset$.
Thus in this case the set $S$ is a proper subset of $T$ and the edge $e$
is a bridge in $\G[T]$. Hence Procedure~$\localalg(G, J)$ returns 
a non-empty set in Line~\ref{ll:bridge}.

Now consider the last case, namely assume that there does \emph{not} exist a 
\tscc induced by a set of vertices $T$ in $\good{\G}{\jj}$ in $\subG{\G}{\jj}$.
Then all vertices in $\good{\G}{\jj}$ are reachable from vertices of $\bad{\G}{\jj}$.
Thus the edge~$e$ is an edge-dominator in~$\subGe{\G}{\jj}(\roote{\G}{\jj})$ by 
Corollary~\ref{cor:kfind} and the procedure returns a non-empty set in 
Line~\ref{ll:dom}.
\end{proof}

In the runtime analysis we need that the algorithm indeed starts local searches 
in each top or bottom \scc or almost top or bottom \scc in the current graph $G$
that was not identified before. For this we use the following observation.
\begin{observation}\label{obs:local} 
Let $G$ be a directed graph. Let $X$ be a set of edges in $G$ and let $J$ be the
set of vertices adjacent to an edge in $X$ in $G$. Let $H$ be a subgraph of~$G$.
Every \tscc in $H \setminus X$ that has incoming edges in~$H$ 
contains a vertex of~$J$.
\end{observation}
\begin{proof}
	Let $S$ be a set of vertices that induces a \tscc in $H \setminus X$ that has 
	incoming edges in $H$. Then the set of edges $X$ has to contain
	an incoming edge $(u,v)$ for some vertex $v \in S$. We have that $v$ is in $J$.
\end{proof}

\begin{theorem}[Runtime]
	Algorithm~\ref{alge2} can be implemented in time $O(m^2 / \log n)$.
\end{theorem}
\begin{proof}
	Let the input graph have $n$ vertices and $m$ edges. Recall that the input graph 
	is converted to a constant degree graph with $n'$ vertices and $m'$ edges such 
	that both $n'$ and $m'$ are of order $O(m)$. Let $G$ denote the (expanded)
	graph maintained by the algorithm.

	Using the linear time algorithms to compute \scc{s} and find all bridges
	in a graph $G$, an iteration of the outer repeat-until loop without the calls
	to $\localalg$ takes time $O(m')$. We will show that there can be only 
	$O(m' / \lost + n' / \depth)$ iterations of the outer 
	repeat-until loop, where $O(m' / \lost + n' / \depth)$ is $O(m' / \log n')$. 
	We will bound the time spent in the inner repeat-until loop separately.
	
	A new iteration of the outer repeat-until loop is started in two cases.
	
	{Case~1}: $\lvert J \rvert \ge \lost$. A vertex is in $J$ only when
	one of its adjacent edges was deleted from $G$ since the last time $J$ was 
	initialized with the empty set. Thus Case~1 can happen at most 
	$2m' / \lost$ times.
	
	{Case~2}: $\localalg$ returned an empty set. Let $G$ and $J$ be
	as maintained by the algorithm at the beginning of the subsequent iteration of 
	the outer repeat-until loop. We distinguish two subcases. Let a subgraph~$G[W]$ 
	induced by some set of vertices $W$ be \emph{connected} if for every partition 
	of~$W$ into two subsets there are edges between the subsets.
	
	Case~\enum{2a}: There exists a subgraph in~$G$ that is connected but not 
	strongly connected. Let $W$ be a set of vertices that induces a maximal
	connected but not strongly connected subgraph of $G$. Since $G[W]$ is 
	connected, the vertices in $W$ were strongly connected when the edges 
	between \scc{s} were removed at the 
	beginning of the previous iteration of the outer repeat-until loop. As $G[W]$ 
	is not strongly connected, it contains a top and a bottom \scc that are
	disjoint.
	By the maximality of $G[W]$, this top and this bottom \scc in~$G[W]$ are also a top
	and a bottom \scc in~$G$, respectively. By Observation~\ref{obs:local} each of them 
	contains a vertex of~$J$. Further, by Lemma~\ref{lem:newtsccball} each of them 
	has more than $\depth$ vertices as otherwise at least one of them would 
	have been identified by Procedure~$\localalg$ and thus they would not be 
	connected in~$G$.
	Both of them are identified when the \scc{s} of $G$ are determined and at least 
	the outgoing edges of the \tscc and the incoming edges of the \bscc are removed
	from~$G$. Thus in this case two vertex sets that each contain more than $\depth$
	vertices are separated from each other by deleting all edges between them.
	This can happen at most $n' / \depth$ times.
	
	Case~\enum{2b}: Every connected subgraph of $G$ is also strongly connected.
	After determining the \scc{s} of~$G$, the algorithm searches 
	for bridges in every \scc of $G$. If no 
	bridge is found, the algorithm terminates. Assume that a bridge~$e$
	was found in an \scc $G[W]$ induced by some set of vertices $W$. The edge~$e$ 
	cannot have been a bridge when 
	bridges were identified in the previous iteration of the outer repeat-until
	loop, as otherwise it would not be in $G$. Since~$e$ is a bridge, there 
	exist an almost top and an almost bottom \scc with respect to~$e$ in~$G[W]$. As every connected subgraph of $G$ is also
	strongly connected, there are no edges between $W$ and $V \setminus W$ in $G$.
	Thus the almost top and the almost bottom \scc in~$G[W]$ are also an almost top and 
	an almost bottom \scc in~$G$. By Observation~\ref{obs:local} each of them 
	contains a vertex of~$J$. By Lemma~\ref{lem:newtsccball} each of them 
	has more than $\depth$ vertices as otherwise the edge~$e$ would have been 
	identified by Procedure~$\localalg$ and removed from the graph. Thus in
	Case~\enum{2b} either the 
	algorithm terminates or two vertex sets that each contain more than $\depth$
	vertices and were strongly connected to each other are separated from each 
	other by deleting the edge~$e$ (such that they are no longer strongly 
	connected but they might still be connected). 
	This can happen at most $n' / \depth$ times.
	
	It remains to bound the time spent in $\localalg$. To this end note that
	each time \emph{before} $\localalg$ is called either \enum{a} $\localalg$ 
	was called and a set of vertices~$S$ that induces a top or bottom \scc~$G[S]$ 
	or an almost top or bottom \scc $G[S]$ in $G$ with respect to 
	some edge $e$ was identified 
	and separated from the remaining graph by deleting the edges between $S$ and 
	$V \setminus S$ or \enum{b} $\allbridges$ identified a bridge~$e$ and 
	increased the number of \scc{s} in~$G$ by removing~$e$.
	Both~\enum{a} and~\enum{b} can happen at most $n'$~times.
	
	We now consider the time for one call to $\localalg$. For each $\jj \in J$
	this procedure runs a breadth-first search of depth~$\depth$ on 
	each of $\G \in \set{G, \revG{G}}$ to identify the subgraphs $\subG{\G}{\jj}$.
	Considering $G$ and $\revG{G}$ only increases the running time by a factor of two.
	The number of edges explored by a breadth-first-search of depth~$\depth$ on a graph 
	with out-degree at most three is~$O(3^\depth)$. Thus with 
	$\depth = \lceil \varepsilon \log n' \rceil$ for some $0 < \varepsilon < 1$ we have 
	that the number of edges in~$\subG{\G}{\jj}$ is $O((n')^\varepsilon)$.
	$\localalg$ computes \scc{s} and bridges or edge-dominators in~$\subG{\G}{\jj}$.
	This can be done in time linear in the number of edges in~$\subG{G}{\jj}$, i.e., in
	time~$O((n')^\varepsilon)$. Thus with $\lvert J \rvert < \lost$ we obtain a time 
	bound of~$O(\lost \cdot (n')^\varepsilon)$ for one call to $\localalg$. Hence the total
	time spent in $\localalg$ can be bounded with 
	$O(\lost \cdot (n')^\varepsilon \cdot n') = O((n')^{1+\varepsilon} \log n')$.
	We have that $O((n')^{1+\varepsilon} \log n')$ is $O(m^2 / \log n)$ for any 
	$\epsilon \in (0,1)$.
\end{proof}

\section{Remark on relation to 2-edge strongly connected blocks}\label{sec:2escbexample}
The following construction shows that, in general, 2-edge strongly connected blocks
do not provide any information about the 2-edge strongly connected components
of a graph. Let $G = (V, E)$ be an arbitrary directed graph. We construct a graph~$G'$
by adding $O(\lvert V \rvert)$ edges and a constant number of vertices to $G$ 
such that all vertices 
in $V$ are in the same 2-edge strongly connected block in $G'$, while the 2-edge 
strongly connected components in $G'$ of the vertices in $V$ remain the same as
in $G$. 
To construct $G'$, we add to $G$: the four vertices $s_1$, $t_1$, $s_2$, 
and $t_2$, the two edges $(s_1, t_1)$ and $(s_2, t_2)$, and for each vertex~$v 
\in V$ the edges $(v, s_1)$, $(v, s_2)$, $(t_1, v)$, $(t_2, v)$. In $G'$ each
vertex $u \in V$ has two edge-disjoint paths to each vertex $v \in V$, namely
the paths $(u, s_1, t_1, v)$ and $(u, s_2, t_2, v)$. Thus all vertices of $V$ are
in the same 2-edge strongly connected block in $G'$. However, clearly
the edges $(s_1, t_1)$ and $(s_2, t_2)$ are bridges in $G'$. Thus the 2-edges
strongly connected components in $G'$ are the same as in $G$ (plus the trivial
subgraphs induced by each of the newly added vertices).

\begin{figure}[h]
\centering
\begin{tikzpicture}
\tikzstyle{graph}=[circle,draw,minimum size=2cm]
\tikzstyle{vertex}=[circle,draw,fill,inner sep=1.2pt, solid]
\tikzstyle{arrow}=[->,line width=0.5pt,>=stealth',shorten >=1pt,thick]
\node[graph] (g) at (0,0) [label=above:{$G$}] {};
\node[vertex] (v) at (0.2,0.6) [label=below:{$v$}] {};
\node[vertex] (u) at (-0.2,-0.3) [label=below:{$u$}] {};
\node[vertex] (s1) at (-2,0.3) [label=left:{$s_1$}] {};
\node[vertex] (s2) at (2,0.3) [label=right:{$s_2$}] {};
\node[vertex] (t1) at (-2,-0.6) [label=left:{$t_1$}] {};
\node[vertex] (t2) at (2,-0.6) [label=right:{$t_2$}] {};

\path (s1) edge[arrow] (t1)
(s2) edge[arrow] (t2)
(v) edge[arrow] (s1) edge[arrow] (s2)
(u) edge[arrow] (s1) edge[arrow] (s2)
(t1) edge[arrow] (v) edge[arrow] (u)
(t2) edge[arrow] (v) edge[arrow] (u)
;
\end{tikzpicture}
\end{figure}}
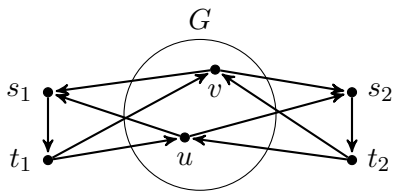
\end{document}